\documentclass[a4paper,twocolumn,11pt,accepted=2023-06-12]{quantumarticle}
\pdfoutput=1
\usepackage{graphicx}
\usepackage{dcolumn}
\usepackage{bm}
\pdfoutput=1
\usepackage[utf8]{inputenc}
\usepackage[T1]{fontenc}
\usepackage{geometry}
\usepackage{etoolbox}

\usepackage{amssymb,amsmath,amscd}
\usepackage{amsmath, amsthm, amscd, amsfonts, amssymb,mathtools}
\usepackage{ulem}
\usepackage{bbm}
\usepackage[dvipsnames]{xcolor}

\definecolor{myurlcolor}{rgb}{0,0,0.7}
\definecolor{myrefcolor}{rgb}{0.8,0,0}
\usepackage{xurl}
\usepackage[unicode=true,pdfusetitle,  bookmarks=true,bookmarksnumbered=false,bookmarksopen=false, breaklinks=false,pdfborder={0 0 0},backref=false,colorlinks=true, linkcolor=myrefcolor,citecolor=myurlcolor,urlcolor=myurlcolor]{hyperref}

\newtheorem{thm}{Theorem}
\newtheorem{lemma}{Lemma}
\newtheorem{rem}{Remark}
\newtheorem{cor}{Corollary}
\newtheorem{prop}{Proposition}
\newtheorem{defn}{Definition}

\renewcommand{\emph}[1]{{\it #1}}
\newcommand{\tr}{\mathrm{Tr }}

\newcommand{\ket}[1]{| {#1} \rangle}

\newcommand{\vect}[1]{{\boldsymbol{#1}}}

\newcommand{\nStab}[3]{\mathcal{C}[[#1, #2]]_{ #3}} 
\newcommand{\Tnkd}{\mathrm{T}(2n, k, \mathbb{Z}_d)} 
\newcommand{\Tnk}[1]{\mathrm{T}(2n, k, \mathbb{Z}_{#1})} 
\newcommand{\Spnd}[2]{\mathrm{Sp}(#1, \mathbb{Z}_{#2})} 
\newcommand{\Sp}{\Spnd{2n}{d}} 
\newcommand{\GLnd}[1]{\mathrm{GL}(#1, \mathbb{Z}_{d})}  
\newcommand{\GL}[2]{\mathrm{GL}(#1, \mathbb{Z}_{#2})} 
\newcommand{\BS}{B_S(n,k,d)} 
\newcommand{\BA}{B_A(n,k,d)} 

\newcommand{\unity}{\mathbbm{1}} 

\newcommand{\Ps}{P^{(d)}} 
\newcommand{\Pn}{\mathcal{P}^{(d)}_n} 
\newcommand{\Ppm}{\mathcal{P}^{(p^m)}} 
\newcommand{\Pf}{\mathcal{P}^{(\mathbb{F}_{p^m})}} 
\newcommand{\Pl}{\mathcal{L}^{(d)}_n} 
\newcommand{\Cl}{\mathcal{C}^{(d)}_n} 

\newcommand{\C}{\mathbb{C}} 
\newcommand{\Z}{\mathbb{Z}_{d}} 
\newcommand{\Zj}[1]{\mathbb{Z}_{p^{#1}}} 
\newcommand{\Zn}{\mathbb{Z}_{d}^n} 
\newcommand{\Znn}{\mathbb{Z}_{d}^{2n}} 
\newcommand{\F}{\mathbb{F}_p} 

\newcommand{\X}{X} 
\newcommand{\ZZ}{Z} 
\renewcommand{\L}{\Lambda} 
\newcommand{\He}{H_{\mathrm{ext}}} 
\newcommand{\M}{\mathbb{M}_{2n}\left(\Z\right)} 
\newcommand{\GM}{\tilde{\Gamma}}
\newcommand{\Mtn}[1]{\mathbb{M}_{2n}\left( \mathbb{Z}_{p_{#1}^{m_{#1}}} \right)} 

\renewenvironment{psmallmatrix} 
  {\left(\begin{smallmatrix}}
  {\end{smallmatrix}\right)} 
  
\newcommand{\CS}{\mathcal{C}(S)} 
\newcommand{\CSp}{\mathcal{C}_{\mathrm{space}}[[n,k]]_d} 
  \newcommand{\CSt}{\mathcal{C}_{\mathrm{state}}[[n]]_d} 
  \newcommand{\CStw}{\mathcal{C}_{\mathrm{state}}[[n]]_2}  
  
  \newcommand{\g}{g} 

\newcommand{\tanmayc}[1]{{ #1}}

\newcommand\SmallMatrix[1]{{%
  \tiny\arraycolsep=0.3\arraycolsep\ensuremath{\begin{pmatrix}#1\end{pmatrix}}}} 

\newcommand{\rvline}{\hspace*{-\arraycolsep}\vline\hspace*{-\arraycolsep}} 

\begin{document}

\title{Counting stabilizer codes for arbitrary dimension}

\author{Tanmay Singal}
\email{tanmaysingal@gmail.com}
\orcid{0000-0001-5986-248X}
\affiliation{Institute of Physics, Faculty of Physics, Astronomy and Informatics,
Nicolaus Copernicus University, \\ Grudzi\c{a}dzka 5/7, 87-100 Toru\'n, Poland}
\author{Che Chiang}
\email{cchiang314@gmail.com}
\affiliation{Department of Physics and Center for Theoretical Physics, National Taiwan University, Taipei 10617, Taiwan}
\author{Eugene Hsu}
\email{a13579230@gmail.com}
\affiliation{Quantum information center, Chung Yuan Christian University, No. 200, Zhongbei Rd., Zhongli Dist., Taoyuan City 320314, Taiwan}
\author{Eunsang Kim}
\email{eskim@hanyang.ac.kr}
\affiliation{Department of mathematical Data Science, Hanyang University, Ansan, Gyeonggi-do, 15588, Korea}
\author{Hsi-Sheng Goan}
\email{goan@phys.ntu.edu.tw}
\affiliation{Department of Physics and Center for Theoretical Physics, National Taiwan University, Taipei 10617, Taiwan}
\affiliation{Center for Quantum Science and Engineering, National Taiwan University, Taipei 10617, Taiwan}
\affiliation{Physics Division, National Center for Theoretical Sciences, Taipei, 10617, Taiwan}
\author{Min-Hsiu Hsieh}
\email{min-hsiu.hsieh@foxconn.com}
\orcid{0000-0002-3396-8427}
\affiliation{Hon Hai Quantum Computing Research Center, Taipei, Taiwan}

\begin{abstract}
 In this work, we compute the number of $[[n,k]]_d$ stabilizer codes made up of $d$-dimensional qudits, for arbitrary positive integers $d$. In a seminal work by Gross (Ref. \cite{Gross2006}) the number of $[[n,k]]_d$ stabilizer codes was computed for the case when $d$ is a prime (or the power of a prime, i.e., $d=p^m$, but when the qudits are Galois-qudits). The proof in Ref. \cite{Gross2006} is inapplicable to the non-prime case. For our proof, we introduce a group structure to $[[n,k]]_d$ codes, and use this in conjunction with the Chinese remainder theorem to count the number of $[[n,k]]_d$ codes. Our work overlaps with Ref. \cite{Gross2006} when $d$ is a prime and in this case our results match exactly, but the results differ for the more generic case. Despite that, the overall order of magnitude of the number of stabilizer codes scales agnostic of whether the dimension is prime or non-prime. This is surprising since the method employed to count the number of stabilizer states (or more generally stabilizer codes) depends on whether $d$ is prime or not.  The cardinality of stabilizer states, which was so far known only for the prime-dimensional case (and the Galois qudit prime-power dimensional case) plays an important role as a quantifier in many topics in quantum computing. Salient among these are the resource theory of magic, design theory, de Finetti theorem for stabilizer states, the study and optimisation of the classical simulability of Clifford circuits, the study of quantum contextuality of small-dimensional systems and the study of Wigner-functions. Our work makes available this quantifier for the generic case, and thus is an important step needed to place results for quantum computing with non-prime dimensional quantum systems on the same pedestal as prime-dimensional systems.
\end{abstract}

\section{Introduction}
\label{sec:introduction}
The stabilizer formalism (which will be explained in detail in Section \ref{sec:pre}) has become an indispensable part of the study of quantum computing. Some of its salient applications are as follows: it forms the bedrock for the vast field of quantum error correcting codes (QECC)\cite{Gottesman97}. Randomising over stabilizer states or Clifford unitaries (see Eq. \eqref{eq:Cliff} for the definition of a Clifford unitary) serves as an important application, for instance, for computing capacities of quantum channels\cite{Bennet96}, data hiding \cite{DiVincenzo2002} and the study of noise-compounding in quantum circuits (randomized benchmarking) \cite{Knill2008, Magesan2011}. It also demarcates a boundary between the classical simulability of quantum computations via the Gottesmann-Knill theorem \cite{Gottesmann99}, and the onset of quantum complexity for which some non-stabilizerness is necessary. The non-universal nature of stabilizer operations \cite{Aaronson2004} and the need to compensate this deficiency with non-stabilizer operations manifests in a resource theory of ``non-stabilizerness", a.k.a., \emph{magic}  (see Ref. \cite{Veitch12, Veitch14, Howard17}), where stabilizer operations are the free resource (see Subsection \ref{subsec:resource} for an explanation of stabilizer operations and the resource theory with respect to which it is defined). While the basic mathematical preliminaries which support all the above have been developed in many works (for e.g. see Ref. \cite{Vourdas2004}, \cite{Gross2006}), there remain some gaps, particularly for quantum computing with multiqudit systems for arbitrary $d$. And while qubit systems are envisaged as the paradigmatic building blocks of quantum computing, qudits of larger dimensions may offer their own benefits (for e.g. see Ref. \cite{Cianci2013,Anwar2014,Campbell2012,Campbell2014,Krishna2019,Heyfron2019}), owing to which it is imperative to study them. Among these, it is easier to study multiqudit systems when the dimension $d$ is a prime number, or when $d$ is the power of a prime, i.e., $d=p^m$, while, simultaneously, the configuration space of the qudit system is the Galois field $\mathbb{F}_{p^m}$. We call such a qudit a Galois-qudit \cite{Galois_qudit}. Many important results which were obtained for qubits may be easily generalised to such qudit systems (this has been noted in many places, for e.g. in the introduction in Ref. \cite{Dangniam20}). This contrasts with qudit systems, whose configuration space is $\Z$, when $d$ is not a prime number. We refer to such qudit systems simply as qudits since this work is about such qudits, but wherever disambiguation is required, we may instead refer to them as modular qudits\cite{modular_qudit}. While there has been a lot of seminal work on such modular qudit systems\cite{Weyl1932,Schwinger1960,Appleby2005,Hostens2005,Gross2006, Bullock2007, Bombin2007, Appleby2012, Gheorghiu2014,Bengtsson2017}, some significant gaps remain. For instance, one of those gaps is the computation of the number of $[[n,k]]_d$ stabilizer codes of such qudits (stabilizer codes are defined in Subsection \ref{subsec:stabilizer}). We will label this number as $\nStab{n}{k}{d}$ \tanmayc{(see Remark \ref{rem:distance} for disambiguation on the usage of the symbol $d$)}. $\nStab{n}{k}{2}$ was computed in a seminal paper Ref. \cite{Aaronson2004} (a recent proof also appeared in Ref. \cite{Lovitz22}), and for Galois-qudit systems, $\nStab{n}{k}{\mathbb{F}_{p^m}}$ was computed in another seminal paper, Ref. \cite{Gross2006}. In this work we compute $\nStab{n}{k}{d}$ for arbitrary $d \ge 2$. Our work overlaps with earlier works\cite{Aaronson2004, Gross2006, Lovitz22} for the case when $d$ is a prime number, and in that case our results match exactly. 
Typically one is interested in the order of magnitude of $\nStab{n}{k}{d}$ since it may act as an important quantifier. For such a reader, we present here the main result.
\begin{cor}[Corollary \ref{cor:nkd_count_order}]
\label{cor:nkd_count_order_intro}
 Let $c=2.17$. Then \tanmayc{the} number of $[[n,k]]_d$ stabilizer codes scales as
\begin{equation}
    \label{eq:nkd_count_order_intro}
   d^{\frac{(n-k)\left(n+3k+1\right)}{2}} \ \le \nStab{n}{k}{d} \ < \ \  d^{\frac{(n-k)\left(n+3k+1\right)}{2}+c},
\end{equation} the number of $[[n,k]]_d$ stabilizer code spaces scale as 
\begin{equation}
    \label{eq:nkd_count_order_code_space_intro}
   d^{\frac{(n-k)\left(n+3k+3\right)}{2}} \ \le \CSp \ < \ \  d^{\frac{(n-k)\left(n+3k+3\right)}{2}+c},
\end{equation} and thus the number of stabilizer states scale as 
\begin{equation}
    \label{eq:nkd_count_order_states_intro}
   d^{\frac{n(n+3)}{2}} \ \le \CSt \ < \ \  d^{\frac{n(n+3)}{2}+c}.
\end{equation}
\end{cor}
Here we make a distinction between an $[[n,k]]_d$ \emph{stabilizer code}, and the corresponding \emph{stabilizer code spaces}, which is one of the $d^{n-k}$ subspaces associated with such a stabilizer code.
\tanmayc{\begin{rem}
\label{rem:distance} One often encounters the phrase $[[n,k,d]]$ code, wherein parameter $d$ stands for the \emph{distance} of the code. In our usage, $d$ refer to the dimension of the Hilbert space of a single qudit, and not the distance of the code.
\end{rem}}

The number of stabilizer states, i.e., $\CSt$, plays the role of an important quantifier in quantum computing. Before this work, $\CSt$ was known only for prime $d$ (and for $\mathbb{F}_{p^m}$ number systems), and, thus, its application as a quantifier was limited to these cases. We list below the salient topics where it has been used as a quantifier before.  
\subsection{Resource theory of magic}
\label{subsec:resource} To initiate the interested reader into the resource theory of stabilizer-operations we refer them to a short summary in Appendix \ref{app:resource}. For more comprehensive treatments on the topic, see Ref. \cite{Veitch14,Howard17,Heinrich19}. For some recent interesting developments in the topic, we further refer the reader to Ref. \cite{Heimendahl20}.
The most significant application of $\CSt$ is that stabilizer states are the extremal points of the stabilizer polytope, and many measures of magic are defined as optimizations over this polytope. Since $\CSt$ scales super-exponentially for prime-dimensional systems (and Galois-qudit systems), computing these measures of magic is an intractable problem for such systems. The earliest known measure in this category is the relative entropy of magic\footnote{In Ref. \cite{Veitch14}, a computable measure of magic called \emph{mana} was also introduced. This measure, while efficiently computable, is computable for only odd-dimensional quantum systems. That being said it has a more significant operational interpretation compared to the relative entropy of magic.} \cite{Veitch14}, after which the robustness of magic was defined in Ref. \cite{Howard17}. The robustness of magic has an operational interpretation: it gives an upper bound on the classical simulation complexity of the quantum Clifford$+\pi/8$ circuits (see also Ref. \cite{Pashayan15}), and scales exponentially in the number of $\pi/8$ gates. It may also be used for the optimality of $\pi/8$-gate counts in the gate-synthesis problem. In Ref. \cite{Heinrich19} the regularised robustness of magic was introduced, and its computation, while superpolynomially faster, still suffers from an exponential run-time. We also refer the reader to a recent work Ref. \cite{Hakkaku21} which tries to optimize the classical simulability of Clifford$+\pi/8$ circuits, in the manner introduced by Ref. \cite{Pashayan15}. The intractability of these measures has also encouraged the search for new quantifiers of magic, for instance, see Ref. \cite{Dai22}. 
Before our work, it wasn't known how $\CSt$ scales with $n$ for arbitrary $d$. Since Corollary \ref{cor:nkd_count_order_intro} states that $\CSt$ scales super-exponentially in $n$ for all $d$, we now know that the problem of computing the aforementioned measures of magic is intractable for arbitrary dimension.
\tanmayc{\subsection{Classical simulation of stabilizer-only operations}
\label{subsec:stim}
The resource theoretic perspective of Subsection \ref{subsec:resource} deems stabilizer operations to be a free resource. The freeness of this resource comes into question if the classical simulation of stabilizer operations has appreciable costs. It is with the aim of minimizing this simulation cost that a qubit stabilizer simulator such as STIM (see Ref. \cite{STIM}) was developed. Any such simulation requires as a pre-requisite adequate memory to store the stabilizer state. Ref. \cite{Aaronson2004,Gross2006} inform us that this is $\mathcal{O}\left(n^2\right)$ bits for qubit systems. Our work supplies this knowledge for qudit systems for general $d$.
}
\subsection{Frame potentials for stabilizer states} 

\label{subsec:frame} In Ref. \cite{Kueng15}, Kueng and Gross established that a uniform ensemble of multiqubit stabilizer states are complex projective $3$-designs. To prove this, they computed explicitly the frame-potential associated with this ensemble of states. And this frame potential was computed using $\CStw$, which was borrowed from Ref. \cite{Aaronson2004} and Ref. \cite{Gross2006}. While for $d \ge 3$, it is known that multi-qudit stabilizer states aren't projective $3$-designs, one may nevertheless still be interested in the projective-design which this ensemble gives rise to. While the result in Ref. \cite{Gross2006} may be employed for this purpose for only prime values of $d$, our results allows one to obtain the result for arbitrary values of $d$.
\subsection{de Finetti theorem for stabilizer states} 
\label{subsec:de_finnetti}
In Ref. \cite{Gross2017} a de Finetti theorem for stabilizer states was established for quantum states on $t$-copes of $n$-qudit systems, where $d$ is prime. To summarize it, let us consider the action of the $t$-th tensor power of the $n$-qudit Clifford group on $t$-copies of an $n$-qudit system, i.e., $U^{\otimes t}$, where $U \in \Cl$, and let $\mathcal{L}$ be the commutant of this $t$th tensor power of the Clifford group, i.e., $L \in \mathcal{L} $ implies that $U^{\otimes t} L {\left( U^\dag\right)}^{\otimes t} = L$. Let a quantum state $\rho$ commute with all $L$ in $\mathcal{L}$. Then $\rho$ has the following property. Define $\rho_s = \tr_{i_1, i_2, \cdots i_{t-s}} \rho$ be the reduced state, which is obtained by tracing out $t-s$ subsystems from $\rho$. The subsystems which are traced out are arbitrary. Then there exists some probability distribution $p_{\sigma_S}$ on the $n$-qudit stabilizer states $\sigma_S$ such that \begin{equation}
    \label{eq:definnetti_1}
    \dfrac{1}{2} \left| \left|  \ \rho_s \ - \ \sum_{\sigma_S} \ p_{\sigma_S} \ \sigma_S^{\otimes s} \ \right|\right|_1 \ \le \ \mathrm{C} d^{\mathcal{O}\left(n^2\right)} \ d^{-\frac{1}{2}\left( t-s \right)},
\end{equation} where $C$ is some constant, $\sigma_S$ is a stabilizer state of an $n$-qudit systems and $p_{\sigma_S}$ is a probability distribution on these states. The exponential scaling in the number of traced out subsystems, i.e., $t-s$ contrasts with the ordinary de Finetti result, whose scaling is of the form $\mathcal{O}\left( s/t \right)$ (see Ref. \cite{Christandl07}). The exponential scaling of the stabilizer version of the de Finetti theorem is proved using $\CSt$, which was computed in Ref. \cite{Gross2006} for prime $d$ (and Galois qudits). We anticipate it to be possible for Corollary \ref{cor:nkd_count_order_intro} to provide a similar exponential scaling for the case when $d$ is non-prime.
\subsection{Others}
\label{subsec:others}
We mention some other topics where $\CSt$ for prime $d$ was employed. We hope that our result could find similar applications for non-prime $d$. \begin{itemize}
    \item[(i)] Quantum contextuality: In Ref. \cite{Howard13}, Howard and collaborators study quantum stabilizer states, to see if these states exhibit any quantum contextual correlations. For odd prime $d$, the set of $\CSt$ two-qudit stabilizer states is partitioned into two classes: separable states and entangled states. Each class is separately studied to see if it exhibits state dependent contextuality, and the number of elements within each class plays an important role in this. 
    \item[(ii)] Disambiguation of Wigner functions for odd dimensional qudit systems: While Hudson's theorem for finite dimensional quantum systems was proved in Ref. \cite{Gross2006}, the choice of the Wigner function with which it was proved wasn't disambiguated. In Ref. \cite{Howard15}, Howard establishes that there is a unique choice of the Wigner function which supports Hudson's theorem. The proof of this result employs the explicit formula for $\CSt$. 
\end{itemize} Our work readily generalises results for the resource theory of magic and for the classical simulability of stabilizer-only circuits. For other topics mentioned above, we anticipate that our computation of $\CSt$ will facilitate comparable results for arbitrary $d$-dimensional systems.   
\tanmayc{\begin{rem}
    \label{rem:CSp} In contrast to $\CSt$, it is difficult to find applications for $\CSp$ (for $k \ge 1$). We believe that this does not exclude it altogether from being a quantity of interest for the quantum computing and information community. For instance see the discussion in Ref. \cite{Gogioso22}: the question posed there is whether there is a known way to sample efficiently from the set of all $[[n,k]]_2$ stabilizer groups for arbitrary $k$. Ref. \cite{Gross2006} may be used to give a sampling algorithm for this purpose.
\end{rem}}
The challenge in computing $\nStab{n}{k}{d}$, for non-prime $d$, is that $\Z$ is not a field. In particular, only elements which are co-prime with $d$ will have multiplicative inverses. The phase space of an $n$-modular qudit system is the $2n-$fold Cartesian product of $\Z$, i.e., $\Znn$. We need to treat it like a vector space which is defined over a field. Towards this end, we will borrow some basics terms and definitions from the theory of vector spaces wherever this is appropriate. We stress though, that the jargon we borrow from linear algebra is actually superfluous, and all the necessary concepts can be phrased in terms of more primitive concepts in the theory of abelian groups and its subgroups. Nevertheless, a rigorous justification of viewing $\Znn$ as a vector space and applying concepts and definitions from the theory of vector spaces will be justified in the appendix. 
$\nStab{n}{k}{d}$ is computed in two steps. (i) We introduce a group theoretic structure of $[[n,k]]_d$ stabilizer codes constructed in the following way: we take a $2n \times (n-k)$ check-matrix of a given code, and extend it to a $2n \times 2n$ symplectic matrix. This extension isn't unique, and the set of all such symplectic extensions realizable for a given $[[n,k]]_d$ code forms a unique coset of $\Sp/\Tnkd$, where $\Tnkd$ is a subgroup of matrices in $\Sp$ which forms the coset corresponding to the $[[n,k]]_d$ trivial code. Conversely, to each coset of $\Sp/\Tnkd$ one can attribute a unique $[[n,k]]_d$ stabilizer code, and thus a bijection between $[[n,k]]_d$ stabilizer codes and cosets of $\Sp/\Tnkd$ is established. This tells us that $\nStab{n}{k}{d} \ = \ \left| \Sp \right|/\left| \Tnkd \right|$. (ii) To compute $\left| \Tnkd \right|$, we find a way to decompose each element of $\Tnkd$ into an ordered product of four matrices, each of which belongs to a distinct subgroup of $\Tnkd$. These subgroups overlaps only over the identity element. Thus $\left| \Tnkd \right|$ is the product of the orders of these subgroups. The orders of these subgroups and of $\left| \Sp \right|$ are computed, first for the case when $d=p^m$, i.e., the power of a prime. This corresponds to the case when the configuration space of the qudit is $\mathbb{Z}_{p^m}$ (and not $\mathbb{F}_{p^m}$). Some of the details of this computation, which employ standard tricks for such computations, are relegated to Section \ref{app:sec_sp_order} in the Appendix. This gives us $\nStab{n}{k}{p^m}$. To obtain $\nStab{n}{k}{d}$, we invoke the Chinese remainder theorem, which is a ring isomorphism, as follows: $\Z \simeq \mathbb{Z}_{p_1}^{m_1} \times \mathbb{Z}_{p_2}^{m_2} \times \cdots \times \mathbb{Z}_{p_r}^{m_r}$, where $d=p_1^{m_1}p_2^{m_2}\cdots p_r^{m_r}$ is the prime factorisation of $d$. While the proof method depends on whether $d$ is prime or non-prime, the number of stabilizer codes $\nStab{n}{k}{d}$, the stabilizer code subspaces $\CSp$ and the stabilizer states $\CSt$ scale agnostic of $d$ being prime or non-prime, which is surprising.  \newline \tanmayc{The above strategy to compute $\nStab{n}{k}{d}$ was also used in \cite{Hanlon2005}, but for the case when $\Z$ is a field\footnote{\tanmayc{$\Z$ is denoted by $\mathbb{F}_q$ in \cite{Hanlon2005}. Eq. (7) in \cite{Hanlon2005} equals the result in Theorem (20) in \cite{Gross2006}. See also \cite{Hanlon_comment}.}}. And in this case, we already know simpler methods to compute $\nStab{n}{k}{d}$, for e.g. \cite{Aaronson2004,Gross2006}.} \newline ~ \newline
The paper is organised as follows. In Section \ref{sec:pre} we give the necessary preliminaries of the stabilizer formalism: we will introduce the $n$-qudit Weyl-Heisenberg group (also known as the generalised Pauli group), we define stabilizer codes along with their check matrices and related concepts. In Section \ref{sec:groupstructure}, we explain the group theoretical structure of $[[n,k]]_d$ stabilizer codes, and in Section \ref{sec:decomposition_tnkd} we show how an arbitrary element of $\Tnkd$ may be decomposed into a product of elements from four different subgroups. In Section \ref{sec:cardinality_nkd} we count $\nStab{n}{k}{d}$. We conclude in Section \ref{sec:conclusion}.
\section{Preliminaries}
\label{sec:pre}
Various preliminaries are introduced in this section, along with the notations and explanations of their usage. 
\subsection{The ring $Z_d$}
\label{subsec:Z_d}
Let $\mathbb{Z}$ denote the set of all integers. Let $d$ be an arbitrary positive integer such that $d \ge 2$. $\Z$ then denotes the set of integers modulo $d$: $\Z = \left\{ 0, 1, \cdots, d-1 \right\}$. All arithmetic operations in $\Z$ are carried out modulo $d$. This means that for any $a,b \in \Z$,
\begin{align}
    \label{eq:Z_addition}
    a \  \oplus_d \  b \ = \ \left( a + b \right) \mod \ d \notag \\ 
    a \ \times_d \ b \ = \ \left(  a \cdot b  \right)\mod d,
\end{align} 
where, on the left hand side of Eq. \eqref{eq:Z_addition}, $\oplus_d$ and $\times_d$ represent addition and multiplication operations in $\Z$, and on the right hand side, $+$ and $\cdot$ represent addition and multiplication in $\mathbb{Z}$. From here onward, we will use the standard notations of addition and multiplication, even in $\Z$, and let the context determine where these operations are meant to be performed.  $\Z$ is a cyclic group under addition, with $0$ as the additive identity. When $d$ is a composite number, for $a \in \Z$ to have a multiplicative inverse in $\Z$, $a$ and $d$ in $\mathbb{Z}$ need to be co-prime.  The set of all non-zero coprime elements form the multiplicative group $\Z^\times$. Hence, $\Z$ is a field if and only if $d$ is a prime number.  When $d$ is a prime, we will denote it by $p$ and denote the corresponding field by $\F$.
\subsection{Using $\Zn$ as a vector space}
Let $n$ be a positive integer. The $n$-fold Cartesian product of $\Z$ with itself, $\Z \times \Z \times \cdots \Z$ contains all ordered $n$-tuples $(a_1,a_2,\cdots, a_n)$ where $a_j \in \Z$. Addition in $\Zn$ is defined point-wise: for arbitrary $a,b \in \Zn$, $(a+b)_i = a_i + b_i$, where $(a+b)_i$ denotes the $i$-th component of $a+b$. Similarly, scalar multiplication of $a$ in $\Zn$ with some $\lambda \in \Z$ means $\lambda a \ = \ (\lambda a_1, \lambda a_2, \cdots, \lambda a_n ).$ $\Zn$ satisfies all the axioms which a usual vector space will satisfy. Thus, any $\vect{a}=(a_1,\cdots,a_n)\in\Zn$ will be called a vector. That being said, since $\Z$ is not necessarily a field, $\Zn$ doesn't form a vector space. Technically, $\Zn$ is called a module over a commutative ring \cite{Artin}. The algebra in $\Zn$ is richer than that for vector spaces defined over fields, and one may not blindly generalise results from vector spaces over fields to $\Zn$. But we will be able to {\it{borrow}} the following concepts from linear algebra. \begin{defn}[Linear independence]
\label{defn:LI}
Let $\vect{a_1}, \vect{a_2}, \cdots, \vect{a_m} \in \Zn$ be $m$ vectors. Then they are linearly independent in $\Zn$ if and only if the only solution for the unknowns $x_1$, $x_2$, $\cdots$, $x_m \in \Z$ in the following equation
\begin{equation}
 \label{eq:LI}
 x_1 a_1 + x_2 a_2 + \cdots + x_m a_m \ = \ 0
\end{equation}
is that $x_1$, $x_2$, $\cdots$, $x_m=0$.
\end{defn}
We now list some corollaries which follow from Definition ~ \ref{defn:LI}.
\begin{cor}
\label{Cor:LI}
\begin{enumerate}
    \item[(i)] Any set of $n$ LI vectors forms a {\it basis} for $\Zn$.
    \item[(ii)] If $\left\{ \vect{a_1}, \vect{a_2}, \cdots , \vect{a_m} \right\}$ are LI, then their {\it linear span} generates a {\it subspace} of {\it dimension} m. 
    \item[(iii)] For any $m < n$, any set of $m$ LI vectors may be extended to form a basis for $\Zn$.
    \item[(iv)] If the columns of an $n \times m$ { \it matrix} are LI, then viewed as a {\it linear map} from $\Zn \rightarrow \Z^m$, its {\it range} is of {\it dimension} $m$ and {\it kernel} is of {\it dimension} $n-m$.
\end{enumerate}
\end{cor} Despite the fact that $\Zn$ is not necessarily a vector space, the following terms in Corollary ~\ref{Cor:LI}:  {\it basis}, {\it linear span}, {\it subspace}, {\it dimension}, {\it matrix}, {\it linear map}, {\it range} and {\it  kernel} can be applied and used in the same manner as done when working with a vector space over fields. The reader is further directed to the Appendix ~ \ref{app:module} for a rigorous justification of these terms.
\subsection{Generalised Pauli group on n qudits}
\label{subsec:Pauli}
We introduce the single qudit Pauli group $\Ps$ using its defining representation, which acts on $\C^d$, where $d$ is an arbitrary positive integer greater than or equal to $2$. The construction of the defining representation of $\Ps$ for arbitrary $d$ has been done earlier in many works
 \cite{Weyl1932,Schwinger1960,Appleby2005,Hostens2005,Gross2006, Bullock2007, Bombin2007, Appleby2012, Gheorghiu2014,Bengtsson2017}. Among these, we choose $\cite{Hostens2005}$. First, let us introduce an orthonormal basis for $\C^d$: $\left\{ \ket{j} \right\}_{j=0}^{d-1}$, where the label $j$ is taken from $\Z$. With respect to this basis, define the linear operators $X$ and $Z$ as follows. 
\begin{equation}
\label{eq:GeneralizedPauli}
\begin{aligned}
& \X \ket{j} = \ket{j+1}, \\
& \ZZ \ket{j} = \omega^{j}\ket{j}, \ \forall \ j \in \Z,
\end{aligned}
\end{equation}
where $\omega \coloneqq \exp{2 \pi i / d}$, ($d$-th root of unity). Since $j \in \Z$, it is understood that the operation of addition in $j+1$ is performed in $\Z$. The commutation relations between $X$ and $Z$ are given by
\begin{equation}
    \label{eq:commutation1}
    Z X \ = \ \omega X Z. 
\end{equation} When $d=2$, $X=\begin{psmallmatrix} 0 & 1 \\ 1 & 0 \end{psmallmatrix}$, and $Z=\begin{psmallmatrix} 1 & 0 \\ 0 & -1 \end{psmallmatrix}$, that is, we get the well-known Pauli matrices for the qubit case. 
 When $d$ is odd, $\Ps \coloneqq \left\langle X , Z \right\rangle$. It is easily seen that an arbitrary element takes the form
$\omega^j X^a Z^b$, where $j,a,b = 0, 1, \cdots, d-1$. Group composition is given by
\begin{equation}
    \label{eq:single_qudit_group_composition1}
  \left(  \omega^j X^a Z^b  \right) \ . \left(  \omega^{j'} X^{a'} Z^{b'} \right) \ = \ \omega^{j+j'+ a'b} X^{a+a'} Z^{b+b'},
\end{equation} which is the just the composition rule of the Heisenberg-Weyl group \cite{Gross2006}. The order of each element is at most $d$. For the case when $d$ is even, often it is convenient\footnote{For even $d$, the group $\left\langle X ,Z \right\rangle$ has some undesirable properties: while $X^d= Z^d=\unity$, $\left( \omega^a XZ \right)^d= - \unity$ for all $a = 0,1,\cdots, d-1$, the order of $\omega^a X Z $ is $2d$. It is cumbersome to keep track of which operator has order $d$ and which has order $2d$. By adding $\zeta$ to the group, we can obtain another group element $\zeta X Z$ whose order is $d$. This is preferable to us because, while $\omega^a X Z \in \left\langle \zeta \unity, X Z \right\rangle$ \tanmayc{continues to have order $2d$}, we can always multiply this with $\zeta$ to get an order $d$ element. The important point is that elements like $X$ and $ \zeta X Z$ are put on the same footing. See \cite{Chaturvedi2010} for a treatment of the even $d$ case without the redundancy $\zeta$ phase factor.} to introduce an additional phase factor $\zeta =\omega^{\frac{1}{2}}$, and $\Ps \coloneqq \left\langle \zeta \unity,  X , Z \right\rangle$. Then an arbitrary element in $\Ps$ is $\zeta^j X^{a} Z^{b}$, where $j \in \mathbb{Z}_{2d}$, $a,b \in \Z$. The group composition law of two arbitrary elements is given by 
\begin{equation}
    \label{eq:single_qudit_group_composition2}
   \left(  \zeta^j X^a Z^b  \right) \ .  \left(  \zeta^{j'} X^{a'}Z^{b'} \right) \ = \ \zeta^{j + j' + 2 a' b} X^{a+a'} Z^{b+b'}.
\end{equation} While the order of all group elements is at most $2d$ in the even $d$ case, one may suitably multiply with $\zeta$ to generate another group element with order at most $d$. 
The $n$-qudit Pauli group $\Pn$ is simply the $n$-fold tensor product of the single qudit Pauli group. The Hilbert space it acts on will be denoted by $\mathcal{H}\simeq \left(\C^d\right)^{\otimes n}$. For the odd $d$ case, $\Pn \ \coloneqq \  \left\langle X_j, Z_j \right\rangle_{j=1}^n$, whereas for the even $d$ case, we have that $\Pn \ = \left\langle \ \zeta \unity, X_j, Z_j \right\rangle_{j=1}^n$, where $X_j$ represents the operator with $X$ on the $j$-th qudit and $\unity$ on all the remaining qudits.   Neglecting the phase factors of $\omega^j$ and $\zeta^j$, we may represent an arbitrary element of the $n$-qudit Pauli group as 
\begin{equation}
    \label{eq:arbitrary_n_qudit_pauli_group}
    \g(\vect{a}) = X^{u_{1}}Z^{v_{1}}  \cdots X^{u_{n}}Z^{v_{n}},  
\end{equation} where  $\vect{a} \coloneqq (u, v)^T \in \Znn$ is a $2n$-ordered tuple with entries in $\Z$. It is easily seen that group composition law is given by 
\cite{Hostens2005}
\begin{equation}
\label{eq:pauliproduct}
\g(\vect{a})\g(\vect{b}) = \omega^{\vect{a}^{T}U \vect{b}}\g(\vect{a}+\vect{b}),
\end{equation}
where
\begin{equation}
U = \left[\begin{matrix}0&0\\I_n&0\end{matrix}\right] ,
\end{equation}
and $I_n$ is the $n \times n$ identity matrix.  From Eq. (\ref{eq:pauliproduct}), the commutation relation for the Pauli operators is
\begin{equation}
\label{eq:commutation2}
\g(\vect{a})\g(\vect{b}) = \omega^{-\vect{a}^{T}\L \vect{b}}\g(\vect{b})\g(\vect{a}),
\end{equation}
where
\begin{equation}
\label{eq:Lambda}
\L \coloneqq \left[\begin{matrix}0&I_n\\-I_n&0\end{matrix}\right].
\end{equation}
Eq. (\ref{eq:commutation2}) implies that 
\begin{equation}
\vect{a}^{T}\L \vect{b} = 0
\end{equation}
if and only if the  two Pauli operators $\g(\vect{a})$ and $\g(\vect{b})$ commute with each other. 
The centre of the Pauli group $\Pn$, which is the subgroup of $\Pn$ which commutes with all elements of $\Pn$, is 
\tanmayc{
\begin{align}
    \label{eq:centre_Pn}
     \mathcal{Z}\left( \Pn \right) \ = \ \begin{cases}  \left\langle \omega \unity \right\rangle, \ \mathrm{when} \ d \ \mathrm{is} \ \mathrm{odd}, \\ \left\langle \zeta \unity \right\rangle, \ \mathrm{when} \ d \ \mathrm{is} \ \mathrm{even}. \\   \end{cases}
\end{align}} Thus the factor group $\Pn / \mathcal{Z}\left(\Pn \right) \simeq \Znn$. This is also seen from Eq. \eqref{eq:pauliproduct}, since $\g(\vect{a})\g(\vect{b}) \propto \g(\vect{a+b})$. Thus $\Pn$ is homomorphic to $\Znn$, and one possible homomorphism takes Pauli $\g(\vect{a}) \in \Pn$ to $\vect{a} \in \Znn$. This homomorphism plays a very important role in the formalism of the stabilizer codes. \newline 
The Clifford group $\Cl$ for an $n$-qudit system is defined as the normaliser of the Weyl-Heisenberg group in the unitary group over the $n$-qudit system. \begin{equation}
    \label{eq:Cliff} 
    U \in \Cl \ \Leftrightarrow \ U g U^\dag \ \in \Pn, \ \forall \ g \in \Pn.
    \end{equation}
The Clifford group is homomorphic to the group of $2n \times 2n$ symplectic matrices $\Sp$ over $\Z$. 
\begin{equation}
\label{eq:M_symp}
M \in \Sp \ \Leftrightarrow \ M^T \L M = \L,  
\end{equation} where $\L$ is given in Eq. \eqref{eq:Lambda}.

\subsection{Stabilizer codes on n qudits}
\label{subsec:stabilizer}
The general framework of stabilizer codes was introduced in \cite{Gottesman97}. Among the many references available in the literature, we refer the reader to \cite{Nielsen}, for a beginner friendly introduction to quantum error correction. \newline 
To define an $[[n,k]]_d$ stabilizer code space, we will need to  first define a stabilizer group $S$. Let $g_1, g_2, \cdots, g_{n-k}$ be $n-k$ elements in $\Pn$, with the following properties. (i) they commute with each other, (ii) the order of $S\coloneqq \left\langle g_1,g_2, \cdots, g_{n-k} \right\rangle$ is $|S| = d^{n-k}$. The final condition ensures that any non-trivial product of the $g_j$'s, i.e., $g_1^{x_1}g_2^{x_2}\cdots g_{n-k}^{x_{n-k}}=\unity$ if and only if $x_1,x_2,\cdots, x_{n-k}=0$. Often, one includes another condition, i.e., spectrum of each $g_j$ always contains $+1$. If any $g_j$ doesn't satisfy this condition, one may replace $g_j$ with \tanmayc{$\omega^{a}g_j$ (or $\zeta^{a}g_j$, as appropriate)}, so that $\omega^{a}g_j$ \tanmayc{(or $\zeta^{a}g_j$)} has eigenvalues $+1$. One may associate to $S$, a subspace of the Hilbert space $\CS$, which is defined as 
\begin{equation}
\label{eq:CS}
\mathcal{C}(S)\coloneqq\left\{\ket{\psi} \in \left(\mathbb{C}^d\right)^{\otimes n}\mid g\ket{\psi} =\ket{\psi} \text{ for all }g\in S\right\}. \end{equation} $\dim \CS =  d^{k}$ (see Theorem 1, in \cite{Gheorghiu2014}). It is easily seen that $\CS$ is the unique subspace {\it{stabilized}} by  $S$, and is hence called the stabilizer code space corresponding to $S$. Since $\dim \CS = d^{k}$, $\CS$ encodes $k$ qudits within itself, and thus $\CS$ is said to be an $[[n,k]]_d$ stabilizer code space, where the subscript denotes the dimension $d$ of a single qudit \tanmayc{(see Remark ~\ref{rem:distance})}. 
\begin{rem}\label{rem:stabilizer_state} When $k=0$, $\dim \CS=1$. Since $\CS$ is spanned by a single vector, we refer to that vector as a stabilizer state. \end{rem} \begin{rem}\label{rem:no_encoding} The case when $k=n$ covers the scenario when we're encoding the whole Hilbert space into itself. We are not interested in this scenario.  \end{rem} 
\tanmayc{Note that $g_j = \g(\vect{a_j})$ for some $\vect{a_j} \in \Znn$. Given $\vect{a}_j$'s instead of $g_j$'s, one may construct the $g_j$'s upto an overall phase factor, from the $\vect{a}_j$'s using Eq. \eqref{eq:arbitrary_n_qudit_pauli_group}. For now, we ignore the loss of this phase factor in going from $g_j$ to $\vect{a}_j$, and identify the subgroup $S$ in terms of  the $\vect{a_j}$'s instead of the $g_j$'s \cite{Nielsen}. This is done by arranging the $\vect{a_j}$'s in a $2n \times (n-k)$ matrix, which is called the check matrix of $S$. The role played by these phase factors in constructing $S$ from $\left\{ \vect{a}_j \right\}_{j=1}^{n-k}$ is explained after Eq. \eqref{eq:M_nequalk} below.}  \begin{equation}
\label{eq:checkmatrix}
H = 
\bigg[\begin{array}{c|c|c|c}
\vect{a_1} &\vect{a_2} &\cdots &\vect{a_{n-k}}
\end{array}\bigg].
\end{equation} The aforementioned two conditions on $g_j$'s may be phrased in terms of equivalent conditions on $\vect{a_j}$'s. 
\begin{align}
    \label{eq:condition1}
    g_ig_j = g_j g_i \ \Leftrightarrow \ \vect{a_j}^T \L \vect{a_i} = 0. \\
    \label{eq:condition2}
    |S| = d^{n-k} \ \Leftrightarrow \  \left| \mathrm{span} \left\{ a_j \right\}_{j=1}^{n-k} \right| \ = d^{n-k}. 
\end{align} Condition Eq. \eqref{eq:condition1} follows from Eq. \eqref{eq:commutation2}, and  Condition Eq. \eqref{eq:condition2} is equivalent to the fact that $a_j$'s are linearly independent and span an $(n-k)$-dimensional subspace. 

The $2n \times (n-k)$ matrix $H$ can be extended to a full $2n \times 2n$ symplectic matrix \tanmayc{$S$}. A partial construction of this full $2n \times 2n$ matrix can be found in Ref. \cite{Nielsen}. We refer the reader to Appendix \ref{appendix:symplectic_realization} for such a construction. \newline 
\begin{equation}
    \label{eq:M_matrix}
    \tanmayc{S}  \ = \ \bigg[\begin{array}{c|c|c|c}
E   & L_X   &  H & L_Z  
\end{array}\bigg],
\end{equation}
where $E$ is a $2n \times (n- k)$ matrix, and $L_X$ and $L_Z$ are $2n \times k$ matrices. The columns of $L_X$ and $L_Z$ correspond to the logical $X$ and the logical $Z$ operators respectively, whereas the columns of $E$ may be interpreted as the {\it correctable} Pauli errors. An explanation of this may be found in Section \ref{sec:decomposition_tnkd} in conjunction with Section \ref{app:sec:significance} of the Appendix. When $k=0$, $H$ is a $2n \times n$ matrix, and $M$ takes the following form.  
\begin{equation}
    \label{eq:M_nequalk}
    \tanmayc{S}  \ = \ \bigg[\begin{array}{c|c}
E   &  H  
\end{array},\bigg]
\end{equation} where $E$ is also a $2n \times n$ matrix. 

When we alter the generators $g_j$ of $S$ as $g_j \rightarrow \omega^{a_j} g_j$, for $a_j \in \Z$, then if $a_j \neq 0$ for some $j$, $\omega^{a_j} g_j$ are the generators of another stabilizer group $S'$. By replacing $S$ with $S'$ in Eq. \eqref{eq:CS}, one can associate to $S'$ a unique corresponding $[[n,k]]_d$ stabilizer code space \tanmayc{$\mathcal{C}\left(S'\right)$}. $\CS$ and $\mathcal{C}\left(S'\right)$ are orthogonal because the $+1$ eigenspaces of $g_j$ and $\omega^{a_j} g_j$ are orthogonal  when $a_j \neq 0$. There are $d^{n-k}$ possible choices for $a_j$ (including the choice $a_j=0$ for all $j$), thus we get $d^{n-k}$ orthogonal code spaces. Since each of these code spaces is of dimension $d^k$, the direct sum of all these code spaces is the full Hilbert space. Since $\vect{a_j}$ do not encode the overall phase factor of the Paulis, the entries in the check-matrix don't change when the phase factors $\omega^{a_j}$ are changed in the generators. So the check matrix $H$ and its symplectic extension \tanmayc{$S$} for all the above codes spaces are the same. We will refer to a \emph{code} as the aforementioned collection of code spaces without meaning to distinguish among them. Thus a code is associated to a check matrix without any ambiguity. We will refer to a \emph{code space} $\CS$ as the subspace of the Hilbert space ${\C^d}^{\otimes n}$ associated to a unique stabilizer group $S$. This is as per Eq. \eqref{eq:CS}. 

\section{Group structure of $[[n,k]]_d$ stabilizer codes}
\label{sec:groupstructure}
In this section, we will obtain a group theoretic structure underpinning all $[[n,k]]_d$ stabilizer codes. 
Let $H_1$ be the check matrix of some $[[n,k]]_d$ stabilizer code, and $A$ be an $(n-k)\times(n-k)$ invertible matrix. Then $H_1A$ is the check-matrix of the same stabilizer code. This is because the columns of $H_1A$ are merely (invertible) linear combinations of the columns of $H_1$, i.e., the generators corresponding to the column vectors of $H_1$ may be recovered from the generators corresponding to the column vectors of $H_2$. Define $H_2=H_1A$. Let \tanmayc{$S_1$} and \tanmayc{$S_2$} be a $2n \times 2n$ symplectic matrix obtained by extending $H_1$ and $H_2$ respectively (see Section \ref{appendix:symplectic_realization} of the Appendix). Since $\tanmayc{S_1, S_2} \in \Sp$, there exists a matrix $M$ such that
\begin{equation}
\label{eq:M_equation}
\tanmayc{S_2} \ = \ \tanmayc{S_1} \ M. 
\end{equation}
We first assume that $k \ge 1$. Then it is clear that $M$ needs to have the following form so that $H_2 = H_1 A$.
\begin{equation}
    \label{eq:M_matrix_block}
    M=\left[\begin{matrix}
M_{11} & M_{12} & 0 & M_{14}\\
M_{21} & M_{22} & 0 & M_{24}\\
M_{31} & M_{32} & A & M_{34}\\
M_{41} & M_{42} & 0 & M_{44}
\end{matrix}\right]
\end{equation}
Here $M_{11}$, $M_{31}$, $A$ are $(n-k)\times (n-k)$ matrices, $M_{22}$, $M_{24}$, $M_{42}$ and $M_{44}$ are $k \times k$ sized matrices, $M_{12}$, $M_{14}$, $M_{32}$, $M_{34}$ are $(n-k) \times k$ and $M_{21}$, $M_{41}$ are $k \times (n-k)$ matrices. Additionally $M$ is necessarily a symplectic matrix, since $M=\tanmayc{S_1^{-1}S_2}$, and $\tanmayc{S_1}$, $\tanmayc{S_2}$ are symplectic. Thus $M$ satisfies the symplectic condition Eq. \eqref{eq:M_symp}. 
Substituting $M$ in \tanmayc{Eq. \eqref{eq:M_matrix_block}} into Eq. (\ref{eq:M_symp}), we get the conditions for $M$ to be symplectic:
\begin{align}
&M_{11}=(A^T)^{-1},\label{sycondition0_1}\\
&M_{12}=M_{14}=0_{(n-k)\times k},\label{sycondition0_2}\\
&\left[\begin{matrix}M_{22}&M_{24}\\M_{42}&M_{44}\end{matrix}\right]\in \mathrm{Sp}(2k,\mathbb{Z}_d),\label{sycondition1}\\
& M_{31}A^T - AM_{31}^T \  = \  A \left(  M_{41}^T M_{21}  \ - \  M_{21}^T M_{41} \right) A^T, \label{sycondition2}\\
& M_{32} \ = \ A \left( M_{41}^T M_{22} -M_{21}^T M_{42} \right)  \label{sycondition3}, \\ 
& M_{34} \ = \ A \left( M_{41}^T M_{24} -M_{21}^T M_{44} \right),  \label{sycondition4}
\end{align}
where $\Spnd{2k}{d}$ denotes the group of $2k \times 2k$ symplectic matrices over $\mathbb{Z}_d$. Using Eq. \eqref{sycondition0_1} and Eq. \eqref{sycondition0_2}, $M$ takes the following form
\begin{equation}
\label{eq:M_form}
M = 
\left[\begin{matrix}
(A^T)^{-1}&0&0&0\\
M_{21}&M_{22}&0&M_{24}\\
M_{31}&M_{32}&A&M_{34}\\
M_{41}&M_{42}&0&M_{44}
\end{matrix}\right].
\end{equation} \noindent The derivation of Eq. (\ref{sycondition1})-(\ref{sycondition4}) is not essential for this work, and is hence shifted to the Appendix \ref{appendix:M_sycondition}. \newline
Note that when $k=0$, $M$ in Eq. \eqref{eq:M_form} takes the form
\begin{equation}
    \label{eq:kis0_M_form}
       M=\left[\begin{matrix}
\left(A^T\right)^{-1} & 0 \\
M_{31} & A 
\end{matrix}\right].
\end{equation}

\begin{thm}\label{tnk}
The set of all $2n\times 2n$ symplectic matrices of the form given in Eq. (\ref{eq:M_form}) form a subgroup of $\Spnd{2n}{d}$.
\end{thm}

{\it Proof. } We use the following subgroup test: if $G$ is a group, $H$ a subset of $G$, then $H$ is a subgroup if for all $h,g \in H$, $h^{-1}g \in H$ too. \newline First note that the inverse of a symplectic matrix $M$ is given by $\L^TM^T\L$, since $\L^T M^T \L M = \L^T \L = I$. This gives us 
\begin{equation}\label{eq:Minverse}
M^{-1} \ = \ \left[\begin{matrix}A^T&0&0&0\\
M_{34}^T&M_{44}^T&0&-M_{24}^T\\-M_{31}^T&-M_{41}^T&A^{-1}&M_{21}^T\\-M_{32}^T&-M_{42}^T&0&M_{22}^T\end{matrix}\right].  \end{equation} Thus we see that $M^{-1}$ has the form given by Eq. \eqref{eq:M_form}. Next, consider the composition of two symplectic matrices of the form  Eq. \eqref{eq:M_form}. \begin{widetext}
\begin{align} \label{eq:subgroup_composition} & \left[\begin{matrix}(A^T)^{-1}&0&0&0\\
M_{21}&M_{22}&0&M_{24}\\M_{31}&M_{32}&A&M_{34}\\M_{41}&M_{42}&0&M_{44}\end{matrix}\right] \left[\begin{matrix}(B^T)^{-1}&0&0&0\\
N_{21}&N_{22}&0&N_{24}\\N_{31}&N_{32}&B&N_{34}\\N_{41}&N_{42}&0&M_{44}\end{matrix}\right] \ = \  \left[\begin{matrix}\left((AB)^{T}\right)^{-1}&0&0&0\\
L_{21}&L_{22}&0&L_{24}\\L_{31}&L_{32}&AB&L_{34}\\L_{41}&L_{42}&0&L_{44}\end{matrix}\right]. \end{align} \end{widetext}
Note that the matrix on the RHS of Eq. \eqref{eq:subgroup_composition} is symplectic, since it is a product of two symplectic matrices. Also it is of the form given in Eq. \eqref{eq:M_form}.  \hfill\qed \\

The subgroup given in Theorem \ref{tnk} will be denoted by $\Tnkd$. It encapsulates the degree of freedom with which one may extend the check matrix of the $[[n,k]]_d$ trivial code to a symplectic matrix. \begin{rem} \label{rem:tnkd}\tanmayc{The subgroup $\Tnkd$ has appeared in the quantum computing literature earlier. For instance, in \cite{Maslov17}, $\mathrm{T}(2n,0,2)$ (referred to as the Borel subgroup of $\Sp$, denoted as $\mathcal{B}_n$) is employed to construct the Bruhat decomposition for $\Spnd{2n}{2}$. This decomposition is used to give an asymptotically tight parameterization of arbitrary stabilizer circuits. More recently, it appears in \cite{Jansen22}, as the subgroup of the Clifford group, whose action does not change the figures of merit of distillation protocols, which are studied in \cite{Jansen22}. See also Remark~\ref{rem:normal_trivia}.} \end{rem} The trivial code is the code corresponding to the stabilizer group $S= \left\langle Z_1, Z_2 , \cdots, Z_{n-k} \right\rangle$. At the same time, right multiplying the symplectic matrix of any $[[n,k]]_d$ stabilizer code by an element of $\Tnkd$ yields another another symplectic matrix of the same stabilizer code. We saw this in Eq. \eqref{eq:M_equation}. In the next section we show that $\Tnkd$ decomposes as a product of three subgroups, and highlight the roles played by these three subgroups in quantum error correction for stabilizer codes. 
\section{Decomposition of $\Tnkd$}
\label{sec:decomposition_tnkd}
The matrix $M$ given in Eq. \eqref{eq:M_form} can be decomposed as a product of \tanmayc{four} matrices. 
\begin{equation}
 \label{eq:M_decomposition}
 M \ = \ M_{T} \ M_{E} \ M_S \ M_L, 
\end{equation}
where  \begin{align}
\label{eqM_T}
& M_T \coloneqq \left[ \begin{matrix}
I_{k}&0&0&0\\
0&I_{n-k}&0&0\\
K_{S}& 0&I_{k}&0\\
0&0&0&I_{n-k}
\end{matrix}\right], \\
    \label{eq:M_E}
     & M_E \coloneqq \left[\begin{matrix}
I_{k}&0&0&0\\
N &I_{n-k}&0&0\\
K_{A} & L^T &I_{k}&-N^T\\
L &0&0&I_{n-k}
\end{matrix}\right], \\
\label{eq:M_L}
& M_L \coloneqq \left[\begin{matrix}
I_k&0&0&0\\
0&M_{22}&0&M_{24}\\
0&0&I_{k}&0\\
0&M_{42}&0&M_{44}
\end{matrix}\right] \  \\ 
 \ \mathrm{and} ~ ~ ~ ~ ~ ~ ~ ~ ~  \  &  M_S \coloneqq \left[\begin{matrix}
(A^T)^{-1}&0&0&0\\
0&I_{n-k}&0&0\\
0&0&A&0\\
0&0&0&I_{n-k}
\end{matrix}\right],\end{align} where \begin{align}
    \label{eq:NKL} &  N=M_{21}A^T, \notag \\ & \ K_S = M_{31}A^T + AM_{31}^T, \notag \\ & \ K_A=M_{31}A^T - AM_{31}^T, \notag \\  \mathrm{and} \ \  & L=M_{41}A^T.
\end{align} Furthermore, we have the following.
\begin{cor}
\label{eq:cor_decomposition_uniqueness}
For any $M \in \Tnkd$, the decomposition in Eq. \eqref{eq:M_decomposition} is unique, i.e., there are unique matrices $M_S$, $M_L$, $M_E$ and $M_T$ such that Eq. \eqref{eq:M_decomposition} holds. 
\end{cor}
\begin{proof}
Since $M_S$ and $M_L$ are constructed from some of the blocks of $M$, there is a unique way of obtaining them. Having obtained $M_S$ and $M_L$, $M_T M_E = M \left( M_L M_S \right)^{-1}$. Since $N$ and $L$ are matrix blocks within $M \left( M_L M_S \right)^{-1}$, and $K_A$ is entirely determined by $N$ and $L$), a unique $M_E$ is obtained from $M \left( M_L M_S \right)^{-1}$, which also gives us a unique $M_T$. Hence proved.
\end{proof} ~ \newline We next note that four subsets of matrices of $M_S$, $M_L$, $M_E$ and $M_T$ form four distinct subgroups, \tanmayc{and the intersection of any pair of these subgroups is $\left\{\unity \right\}$}. \newline 
\paragraph{The $\GLnd{n-k}$ subgroup:} The set of matrices of the form $M_S$, with $A$ invertible, is readily seen to be a subgroup of $\Tnkd$ (by setting all matrix elements in $M_{21}$, $M_{24}$, $M_{31}$, $M_{32}$, $M_{34}$, $M_{41}$ and $M_{42}$ as $0$, and setting $M_{22}=M_{44}=I_{n-k}$). Moreover, this subgroup is isomorphic to $\GLnd{n-k}$, the group of $(n-k)\times(n-k)$ general linear matrices over $\Z$, and hence we refer to it as the $\GLnd{n-k}$ subgroup. 
\paragraph{ The $\Spnd{2k}{d}$ subgroup:} All matrices of the form $M_L$ which satisfy Eq. \eqref{sycondition1} form a subgroup of $\Tnkd$, and this subgroup is isomorphic to $\Spnd{2k}{d}$. Thus we call it the \tanmayc{$\Spnd{2k}{d}$} subgroup. \newline
Note that the $\GLnd{n-k}$ subgroup and the \tanmayc{$\Spnd{2k}{d}$} subgroup commute with each other. \newline
\paragraph{The symmetric and antisymmetric abelian subgroups $\BS$ and $\BA$:} Matrices of the form $M_T$ and $M_E$, which satisfy Eq. \eqref{sycondition2} (see Eq. \eqref{eq:NKL}), are also subgroups of $\Tnkd$. For $M_T$ to satisfy Eq. \eqref{sycondition2}, $K_S$ has to be symmetric but is otherwise unconstrained, while for $M_E$ to satisfy Eq. \eqref{sycondition2}, $K_A$ has to be anti-symmetric and moreover satisfies the equation $K_A = N^TL - L^TN$ (see Eq. \eqref{eq:NKL}).  It is easily verified that both are abelian subgroups, and also commute with each other. We refer to them as the symmetric and anti-symmetric abelian subgroups $\BS$, $\BA$, 
\begin{align}
    \label{eq:subgroup_BS}
    & \BS \ = \ \left\{ \ \mathrm{all} \  M_T \in \Tnkd \right\}. \\ 
    \label{eq:subgroup_BA}
    & \BA \ = \ \left\{ \ \mathrm{all} \  M_E \in \Tnkd \right\}.
\end{align}
\begin{rem}
\label{rem:kis0_decomposition}
When $k=0$, the $\BA$ and $\Spnd{2k}{d}$ subgroups shrink to the trivial group which contains only the identity. The decomposition of $\Tnkd$ then takes the following form. \begin{equation} \label{eq:kis0_tnkd_decomposition} 
M  \ = \ M_T \ M_S.  
\end{equation} The uniqueness of the decomposition (i.e., Corollary~ \ref{eq:cor_decomposition_uniqueness}) holds for this case well.
\end{rem} \tanmayc{\begin{rem}\label{rem:normal_trivia} The group $\Tnkd$ is often referred to in the mathematics literature as the parabolic subgroup associated to a $k$-dimensional isotropic subspace of $\Znn$ \cite{Hanlon2005}. Furthermore, there is more structure to it than we mentioned in its decomposition above (see Remark \ref{rem:tnkd_structure}). \end{rem} \begin{rem} \label{rem:tnkd_structure} We borrow this observation from \cite{Hanlon2005}. Let us denote by $B$ the group generated by the union of $\BS$ and $\BA$, i.e., $N \coloneqq \left\langle \BS \cup \BA \right\rangle$. Then it is easily verified that $N$ is a normal subgroup of $\Tnkd$, and that $\Tnkd/ N \simeq \GL{n-k}{d} \times \Spnd{2k}{d}$. \end{rem} } In the Appendix Section \ref{app:sec:significance}, we elaborate on the significance of these subgroups in quantum error correction. While the roles played by $\GLnd{n-k}$ and $\Spnd{2k}{\Z}$ is easily surmised to those familiar with QEC with stabilizer codes, there has been less spotlight on the roles played by $\BA$ and $\BS$, i.e., the fact that the choice of selecting correctable errors for the code is encapsulated by the \emph{group actions} of $\BS$ and $\BA$. 
\section{Counting $[[n,k]]_d$ stabilizer codes}
\label{sec:cardinality_nkd}
One may employ Lagrange's theorem in theory of finite groups to compute the total number of $[[n,k]]_d$ stabilizer codes, which we denote as $\nStab{n}{k}{d}$.
\begin{lemma}\label{lem:count}
The number of $[[n,k]]_d$ stabilizer codes  is 
\begin{equation}
    \label{eq:nkd_stab_code_count}
    \nStab{n}{k}{d} \ = \ \left| \Spnd{2n}{d} \right|/ \left| \Tnkd \right|,
\end{equation} where $|\Spnd{2n}{d}|$ is the order of $\Spnd{2n}{d}$ and $|\Tnkd|$ is the order of $\Tnkd$.
\end{lemma}
\begin{rem}
\label{rem:d_prime_Gross}
For the case when $d$ is a prime, i.e., $d=p$, the number of $[[n,k]]_d$ stabilizer codes was explicitly computed in Ref. \cite{Gross2006} (see Theorem 20 and Corollary 21 therein). The results we obtain below agree with the results in Ref. \cite{Gross2006} in this case. There is a discrepancy in the language and notation employed: an $[[n,k]]_d$ stabilizer code in our work corresponds to an $m$-dimensional isotropic subspace in Theorem 20 in Ref. \cite{Gross2006}, each $[[n,k]]_d$ code space in our work is a counted as a distinct code in Corollary 21 in Ref. \cite{Gross2006}, and finally, we use the notation $\nStab{n}{k}{d}$ to count the number of $[[n,k]]_d$ stabilizer codes, whereas it is used to count the total number of code spaces in Ref. \cite{Gross2006}.
\end{rem}
\begin{proof} For a given $[[n,k]]_d$ stabilizer code, let's choose a representative code space $\CS$ with the stabilizer group $S=\left\langle g_j \right\rangle_{j=1}^{n-k}$. From $S$ one can construct a check matrix $H$, which may be extended to a symplectic matrix $M$. In Sec. \ref{sec:groupstructure} we noted that this construction of $M$ from $S$ has a redundancy, and that the degree of freedom within the redundancy is captured by the right action on $M$ by the subgroup $\Tnkd$. In other words, each symplectic matrix in any coset of $\Sp/\Tnkd$ is a possible symplectic construction for the same $[[n,k]]_d$ stabilizer code. This tells us that the number of $[[n,k]]_d$ codes is lesser than or equal to the number of left cosets, i.e., $\left| \Sp / \Tnkd \right|$.
Conversely, let us start from some left coset in $\Sp / \Tnkd$, i.e., $M \times \Tnkd$, where we designate $M$ as the coset representative. One can associate to this coset any one of $d^{n-k}$ mutually orthogonal $[[n,k]]_d$ stabilizer codes in the following way: extract from $M$ a $2n \times (n-k)$ submatrix by extracting the columns with numbers $n+1$ to $2n-k$ from the left (see Eq. \eqref{eq:M_matrix}). Call this submatrix $H$. Since $M \in \Sp$ and satisfies $M^T \L M=\L$, $H$ satisfies $H^T \L   H = 0$. Let the $j$-th column in $H$ be the homomorphic image of some Pauli $\omega^{a_j} g_j$ (see Eq. \eqref{eq:arbitrary_n_qudit_pauli_group}), where $g_j$'s are such that each has a $+1$ eigenspace, and the $a_j$'s are in $\Z$ and arbitrary. Generate the group $S=\left\langle \omega^{a_j} g_j \right\rangle$. Note that the construction branches out into $d^{n-k}$ different choices of $S$, depending on the values of the $a_j$'s. That $H^T \L H =0$ is equivalent to the fact that the $g_j$'s commute between themselves. Also, the columns of $H$ are linearly independent since $M$ is an invertible matrix. This ensures that any subset of the $\omega^{a_j}g_j$'s will generate a strictly smaller subgroup of $S$. Construct the coding space $\CS$ from $S$, and note that $\dim \CS = 2^k$ (see Eq. \eqref{eq:CS} and the description below it). Thus $\CS$ is a code space corresponding to an $[[n,k]]_d$ stabilizer code. The $d^{n-k}$ distinct choices of $S$ and $\CS$ correspond to the $d^{n-k}$ different choices within the same stabilizer code. This tells us that the number of cosets is lesser than or equal to the number of $[[n,k]]_d$ stabilizer codes. \newline
Note that when the code to coset construction is reversed, one retrieves the original code which one started with. This is true for all codes. This proves the theorem.

\end{proof} Using the decomposition of $\Tnkd$ from Sec. \ref{sec:decomposition_tnkd} we obtain the following theorem. \begin{widetext}
\begin{lemma}
\label{lem:stabilizer_codes}
Let $d$ be a prime power, i.e., $d=p^{m}$, where $p$ is a prime number and $m$ a positive integer. The total number of $[[n,k]]_{p^m}$ stabilizer codes, $\nStab{n}{k}{p^m}$ is 
\begin{align}
\label{eq:count_stab1}
  &   \nStab{n}{k}{p^m} \ = \ \left( p^{m-1} \right)^{\frac{(n-k)(n+3k+1)}{2}} \ \left[\begin{matrix} n \\ n-k \end{matrix}\right]_{p} \ \prod_{j=0}^{n-k-1} \left( \left(p\right)^{n-j} +1  \right) \end{align} where the Gaussian coefficient $\left[\begin{matrix} n \\ n-k \end{matrix}\right]_{p_i}$ is defined as
\begin{equation}
    \label{eq:Gaussian_coefficients}
    \left[\begin{matrix} n \\ n-k \end{matrix}\right]_{p} \ \coloneqq \ \prod_{j=0}^{n-k-1} \ \dfrac{p^{n-j}-1}{p^{n-k-j}-1}.
\end{equation} 
\end{lemma} \end{widetext}
\begin{proof} Using Lemma \ref{lem:count}, the decomposition in Sec. \ref{sec:decomposition_tnkd} and Lemma \ref{eq:cor_decomposition_uniqueness}, \begin{widetext}
\begin{equation}
    \label{eq:count_stab}
    \nStab{n}{k}{p^m} \ = \ \begin{cases}  \  \dfrac{\left| \Spnd{2n}{p^m}  \right|}{\left| B_S\left(n,k,p^m \right) \right| \ \left| B_A\left( n,k,p^m \right) \right| \  \left| \GL{n-k}{p^m}  \right| \  \left| \Spnd{2k}{p^m} \right|}, \ \mathrm{when} \ k \ge 1 \ \mathrm{and}  \\  ~ \\
    \dfrac{\left| \Sp  \right|}{\left| B_S \left( n, 0, p^m \right) \right| \left|\GL{n}{p^m} \right|}, \ \mathrm{when} \ k=0.
    \end{cases} 
\end{equation} \end{widetext} It is straightforward to compute that $\left| B_A\left( n,k,p^m \right)  \right|=\left( p^m \right)^{2k(n-k)}$, since in $M_E$ from Eq. \eqref{eq:M_E}, $N$ and $L$ are left completely unconstrained by Eq. \eqref{eq:NKL}, whereas simultaneously $K_A$ is completely determined by $N$ and $L$. Also note, $\left| B_S\left(n,k,p^m \right) \right|={p^m}^{\frac{1}{2}(n-k)\left(n-k+1\right)}$, since in $M_T$ from Eq. \eqref{eqM_T}, the only constraint which Eq. \eqref{eq:NKL} imposes on $K_S$ is that it be symmetric. $\left| \mathrm{GL}\left(n-k, \mathbb{Z}_{p^{m}} \right) \right|$ is computed in Corollary 2.8 in Ref. \cite{Han} (see also Ref. \cite{Vaidyanathan} for a more accessible arguments): \begin{align}
    \label{eq:order_order_GL}
 &   \left| \mathrm{GL}\left(n-k, \mathbb{Z}_{p^{m}} \right) \right| \notag \\   = &  \ p^{\left(m-1\right)(n-k)^2} \ \prod_{j=0}^{n-k-1}  \ \left( p^{n-k} - p^j \right).  
\end{align} $\left|\Spnd{2n}{p^{m}}\right|$ and $\left| \Spnd{2k}{p_i^{m_i}}\right|$ (needed for the case $k\ge 1$) are explicitly computed in the Appendix Sec. \ref{app:sec_sp_order} (which is based on Ref. \cite{hattice}), which gives us 
 \begin{align}
   \label{eq:order_sp_2} 
   \left| \Spnd{2n}{p^{m}}  \right| \ = \ {p}^{(2m-1)n^2+ (m-1)n} \ \prod_{j=1}^{n}\left( p^{2j}-1\right).
   \end{align} Then putting everything together in Eq. \eqref{eq:count_stab} gives us Eq. \eqref{eq:count_stab1} 
\end{proof}
\begin{rem}
\label{rem:Gross}
In a seminal paper on Wigner distributions on finite dimensional phase-space Ref. \cite{Gross2006}, Gross computed the number of $[[n,k]]_d$ codes when $d=p^m$, but the phase-space is identified with $\mathbb{F}_{p^m}^2$, not $\mathbb{Z}_{p^m}^2$. That the Pauli groups for both phase spaces are different is observed from the following fact. Set $n=1$. Then $\Ppm$ is homomorphic to the additive abelian group of $\mathbb{Z}_{p^m}^{2}$ (i.e., when only looking at $\mathbb{Z}_{p^m}^{2}$ as an abelian group), whereas $\Pf$ is homomorphic to the additive abelian group of $\mathbb{F}_{p^m}^2$. Note that the abelian group of $\mathbb{Z}_{p^m}^{2}$ decomposes as the product of two cyclic groups: $\left(\mathbb{Z}_{p^m}^{2},+\right)\simeq \left(\mathbb{Z}_{p^m},+\right)\times \left(\mathbb{Z}_{p^m},+\right) $ , whereas that of $\mathbb{F}_{p^m}^2$ decomposes as a product of $2m$ cyclic groups, i.e.,  $\left(\mathbb{F}_{p^m}^2,+)\right)\simeq \left(\mathbb{Z}_p,+\right)^{\times 2m}$. Furthermore these decompositions are unique by the structure theorem of finite abelian groups. When $m=1$, i.e., $d=p$, the formula in Theorem 20 in Ref. \cite{Gross2006} matches that with Eq. \eqref{eq:count_stab1}, but for larger $m$, the formula differs on account of the fact that while $\mathbb{F}_{p^m}$ is a field, $\mathbb{Z}_{p^m}$ isn't. For completeness, in Appendix \ref{app:sec:Gross_proof} we explain why this renders the proof for Theorem 20 in Ref. \cite{Gross2006} inapplicable to our case.
\end{rem} 
\begin{thm} \label{thm:counting_general_d} When $d$ is an arbitrary positive integer with prime factorisation $d=\prod_{i=1}^r p_i^{m_i}$, where $p_i$ are distinct primes, $m_i$ are positive integers  and $r$, which is the number of distinct prime factors, is also a positive integer. Then total number of $[[n,k]]_d$ stabilizer codes is
\begin{align} 
\label{eq:count_stab2}
&\nStab{n}{k}{d} \ =  \ \prod_{i=1}^r \ \nStab{n}{k}{p_i^{m_i}}.
\end{align} Thus we get
\begin{equation}
    \label{eq:count_stab2_1}
    \nStab{n}{k}{d} \ =  \ d^{\frac{(n-k)\left(n+3k+1\right)}{2}} \ \  \prod_{i=1}^r   \ \zeta_i,
\end{equation} where 
\begin{equation}
    \label{eq:zeta}
    \zeta_i \ \coloneqq \ \prod_{j=0}^{n-k-1} \left( \dfrac{1-{p_i}^{-2(n-j)}}{1-{p_i}^{-(n-k-j)}} \right).
\end{equation}\end{thm}
\begin{proof}
We invoke the Chinese remainder theorem \cite{Lang}, which tells us that \begin{equation}
    \label{eq:Chinese_remainder_theorem}
    \Z \  \simeq \  \mathbb{Z}_{p_1}^{m_1} \times \mathbb{Z}_{p_2}^{m_2} \times \cdots \mathbb{Z}_{p_r}^{m_r},
\end{equation} is a \emph{ring} isomorphism. A simple corollary of this ring isomorphism is the following two group isomorphisms. \begin{widetext}
\begin{align}
    \label{eq:Chinese_remainder_theorem_cor}
   & \Sp \ \simeq \ \Spnd{2n}{p_1^{m_1}} \times \Spnd{2n}{p_2^{m_2}} \times \cdots \times  \Spnd{2n}{p_r^{m_r}} \notag \\ 
 &   \Tnkd \ \simeq \ \Tnk{p_1^{m_1}} \times \Tnk{p_2^{m_2}} \times \cdots \times \Tnk{p_r^{m_r}},
\end{align} \end{widetext} where the $\times$ symbol on the RHS denotes the direct product of groups. \tanmayc{To explain Eq. \eqref{eq:Chinese_remainder_theorem_cor}, we refer the interested reader to Section \ref{app:chinese_remainder_thm} of the Appendix.} Invoking Eq. \eqref{eq:Chinese_remainder_theorem_cor} into Lemma \ref{lem:count} thus proves Eq. \eqref{eq:count_stab2}. For Eq. \eqref{eq:count_stab2_1}, we first simplify the RHS of Eq. \eqref{eq:count_stab1} using the following.
\begin{align}
    \label{eq:Gaussian_coeff}
   &  \left[\begin{matrix} n \\ n-k \end{matrix}\right]_{p_i} \ \prod_{j=0}^{n-k-1} \left( \left(p_i\right)^{n-j} +1  \right) \ \notag \\ = & \ p_i^{\frac{(n-k)\left(n+3k+1\right)}{2}} \ \prod_{j=0}^{n-k-1} \left( \dfrac{1-{p_i}^{-2(n-j)}}{1-{p_i}^{-(n-k-j)}} \right). 
\end{align} Using Eq. \eqref{eq:Gaussian_coeff} in Eq. \eqref{eq:count_stab1} gives us
\begin{equation}
    \label{eq:count_stab2_2}
    \nStab{n}{k}{p_i^{m_i}} \ = \ {p_i^{m_i}}^{\frac{(n-k)\left(n+3k+1\right)}{2}} \  \zeta_i,
\end{equation} which is then invoked into Eq. \eqref{eq:count_stab2}.
\end{proof}
It is often important to simply get an order of magnitude of $\nStab{n}{k}{d}$. Towards that result, we obtain the following corollary as a result of Theorem \ref{lem:stabilizer_codes}. 
\begin{cor}
\label{cor:nkd_count_order} Let $c=2.17$. Then \tanmayc{the} number of $[[n,k]]_d$ stabilizer codes scales as
\begin{equation}
    \label{eq:nkd_count_order}
   d^{\frac{(n-k)\left(n+3k+1\right)}{2}} \ \le \nStab{n}{k}{d} \ < \ \  d^{\frac{(n-k)\left(n+3k+1\right)}{2}+c},
\end{equation} the number of $[[n,k]]_d$ stabilizer code spaces scale as 
\begin{equation}
    \label{eq:nkd_count_order_code_space}
   d^{\frac{(n-k)\left(n+3k+3\right)}{2}} \ \le \CSp \ < \ \  d^{\frac{(n-k)\left(n+3k+3\right)}{2}+c},
\end{equation} and thus the number of stabilizer states scale as 
\begin{equation}
    \label{eq:nkd_count_order_states}
   d^{\frac{n(n+3)}{2}} \ \le \CSt \ < \ \  d^{\frac{n(n+3)}{2}+c}.
\end{equation}
\end{cor}
\begin{proof}
It is not difficult to upper bound $\zeta_i$ by a constant which is independent of $p_i$, $n$ and $k$. For instance, in Appendix Sec. \ref{app:sec:inequality} we show that  
\begin{equation}
    \label{eq:bound_prod}
    \zeta_i \ < \ e^{1.57}, \ \forall \ n , \  k, \, \mathrm{and}  \ \mathrm{primes} \ p.
\end{equation} Thus we now use Eq. \eqref{eq:Gaussian_coeff} and Eq. \eqref{eq:bound_prod} 
in Eq. \eqref{eq:count_stab1} 
\begin{equation}
    \label{eq:count_stab3} 
    \nStab{n}{k}{p_i^{m_i}} \ < \ e^{c_1} \  \left( p_i^{m_i}\right)^{\frac{(n-k)(n+3k+1)}{2}},
\end{equation} where $c_1 = 1.57$. Let the number of distinct prime factors of $d$ be denoted by $r(d)$. Then Eq. \eqref{eq:count_stab2} and \eqref{eq:count_stab3} tell us 
\begin{equation}
\label{eq:count_stab_4}
\nStab{n}{k}{d} \ < \ e^{c_1   r(d)} \  \ d^{\frac{(n-k)(n+3k+1)}{2}}.
\end{equation}
It is known that $r(d)$ may be upper bounded as follows (see Theorem 11, p. 369 in Ref. \cite{Robin}): 
\begin{equation}
    \label{eq:r_upperbound}
    r(d) \ \le \ c_2  \ \dfrac{\log d}{\log \log d},  \ \mathrm{for} \ d \ge 3. 
\end{equation} where $c_2 = 1.38$. Thus we get that $e^{c_1 r(d)} \ \le \ d^{c}$, where $c=c_1 c_2 \approx 2.17 $. Thus we get 
\begin{equation}
\label{eq:count_stab_5}
\nStab{n}{k}{d} \ < \ d^{\frac{(n-k)(n+3k+1)}{2} + c}.
\end{equation} Note that for $d=2$, Eq. \eqref{eq:count_stab3} is already satisfied, and hence, $d=2$ also satisfies Eq. \eqref{eq:count_stab_5}. The number of $[[n,k]]_d$ stabilizer code spaces is $d^{n-k} \ \nStab{n}{k}{d}$, and the number of stabilizer states is obtained by setting $k=0$ in the number of stabilizer code spaces.
 \end{proof} 
\section{Conclusions}
\label{sec:conclusion}
In this work we count the number of $[[n,k]]_d$ stabilizer codes, $\nStab{n}{k}{d}$ for arbitrary $d$-level systems, where the configuration space of such systems is $\Z$. Since $\Z$ is not a field when $d$ is non-prime, the method we used for this has to differ from earlier works in Ref. \cite{Aaronson2004} for qubits and Ref. \cite{Gross2006} for prime-dimensional qudits and Galois-qudits, which relied on $d$ being prime. Our method is broadly broken up into two parts: (i) proving a bijection between distinct $[[n,k]]_d$ QECC and cosets of $\Spnd{2n}{\Z}/\Tnkd$, where $\Tnkd$ is a subgroup, which corresponds to the $[[n,k]]_d$ trivial code, with stabilizer group $S \ = \ \left\langle Z_j \right\rangle_{j=1}^{n-k}$, and (ii) computing $\left| \Sp \right|/ \left| \Tnkd \right|$, for which we make use of the Chinese remainder theorem. We find that the number of $[[n,k]]_d$ stabilizer codes $\nStab{n}{k}{d}$, the number of $[[n,k]]_d$ stabilizer code subspaces $\CSp$ and the number of stabilizer states $\CSt$ scale agnostic of whether $d$ is prime or non-prime. This is surprising since the prime or non-prime nature of $d$ played an important role in our computation. In Section \ref{sec:introduction} we listed salient topics where $\CSt$ plays the role as an important quantifier: the resource theory of magic, the classical simulation of stabilizer-only circuits, projective-designs of stabilizer states, a de Finetti theorem customized for stabilizer operations, the study of quantum contextuality for small systems \tanmayc{and} the study of Wigner functions. Since $\CSt$ was so far known only for prime $d$ (and Galois-qudits), the results were limited to such cases. For resource theory of magic, our work allows one to make statements for arbitrary $d$. For the remaining topics, we believe that our computation of $\CSt$ will prove useful for generalising the corresponding results to arbitrary $d$. Thus our work also contributes towards the important goal of placing results for arbitrary $d$ qudit systems on the same pedestal as for prime qudit systems or Galois qudit systems. \newline 

\begin{acknowledgments}
We wish to acknowledge help from the referees, whose comments were very useful for us.  H.-S.G. acknowledges support from the National Science and Technology Council, Taiwan under Grants  No.~NSTC 112-2119-M-002-014, No.~NSTC 111-2119-M-002-006-MY3, No.~NSTC 110-2627-M-002-002, and No.~NSTC 111-2627-M-002-001, and from the National Taiwan University under Grant No.~NTU-CC-112L893404. H.-S.G. is also grateful for the support from the “Center for Advanced Computing and Imaging in Biomedicine (NTU-112L900702)” through The Featured Areas Research Center Program within the framework of the Higher Education Sprout Project by the Ministry of Education (MOE), Taiwan, and the support from the Physics Division, National Center for Theoretical Sciences, Taiwan.
\end{acknowledgments}

\appendix
\onecolumngrid
\section{A short summary of the resource theory of magic}
\label{app:resource}
An operationally significant resource theory of magic has been studied and developed over the past decade. In this theory stabilizer operations are considered a free resource. By stabilizer operations we mean the following. (i) Initiating the starting state of a quantum computation as an eigenstate of some multi-qudit Pauli operator. Often this is the state $\ket{0,0,\cdots,0}$, which spans the $[[n,0]]_d$ trivial code for the $n$-qudit system with stabilizer group $S= \left\langle Z_j \right\rangle_{j=1}^n$. (ii) Performing Clifford operations during the computation. (iii) The only observables which are subjected to quantum measurements are multiqudit Pauli operators. It is known that stabilizer operations are not universal \cite{Aaronson2004}, and universality demands the inclusion of some non-stabilizer operations, such as the $\pi/8$ gate via magic state injection. From the perspective of fault-tolerant quantum computing, performing non-stabilizer operations at the \emph{logical} level of the quantum code, is a high-cost incurring operation (see for e.g. \cite{OGorman17}), which justifies the differentiation between the `free' resources of stabilizer operations (the aforementioned three operations) and the costly resources of non-stabilizer operations (magic state injection).
\section{Proof of Corollary \ref{Cor:LI}}
\label{app:module}
$\Zn$ is a free-$\Z$ module. This means that one may simply view $\Zn$ as an abelian group with vector addition as the corresponding group operation. Furthermore, $\Zn$ is then simply seen to be a direct product group of $\Z$ with itself $n$ times.  
\subsection{Any set of $n$ LI  vectors is a basis for $\Zn$}
\label{subsec:app1}
By {\it basis} for $\Zn$, we mean a {\it minimal generating set } for $\Zn$ as a group. Any set of $n$ LI vectors is a minimal generating set since there are precisely $d^n$ distinct linear combinations of these $n$ vectors, and the LI property guarantees that each distinct linear combination results in a distinct vector in $\Zn$. Since there are $d^n$ vectors in $\Zn$, each vector must be one of the $d^n$ possible linear combinations.
\subsection{If $\mathbf{a_j}$'s are LI, then their linear span generates a subspace of dimension $m$}
\label{subsec:app2}
By the linear span of the $\vect{a_j}$'s, we mean the subgroup of $\Zn$, which the $\vect{a_j}$'s generate. Since the $\vect{a_j}$'s are LI and since they are $m$ in number, they generate a subgroup of size $d^m$. However, not all subspaces of size $d^m$ are isomorphic. It remains to justify the concept of an {\it $m$ dimensional subspace}. We use the term {\it $m$-dimensional subspace} to mean a subgroup of $\Zn$, which is isomorphic to $\Z^m$. One may easily construct such an isomorphism from $\mathrm{span} \left\{ \vect{a_j} \right\}_{j=1}^m$ to $\Z^m$ as follows: define $T:\mathrm{span} \left\{ \vect{a_j} \right\}_{j=1}^m \rightarrow \Z^m$ as follows: \begin{equation}
    \label{eq:isomorphism}
    T \left(  x_1 \vect{a_1} + x_2 \vect{a_2} + \cdots + x_m \vect{a_m} \right) \ = \ (x_1,x_2,\cdots, x_m), \ \forall \ x_j \in \Z.
\end{equation}
$T$ is readily seen to be an isomorphism (in the group theoretic sense). 
\subsection{For any $m < n$, any set of $m$ LI vectors may be extended to form a basis for $\Zn$.} 
\label{subsec:app3} We prove this by finding an invertible matrix $A$ whose first $m$ columns are given by the LI vectors $\vect{a}_1$, $\vect{a}_2$, $\cdots$, $\vect{a}_m$. Then we may simply add the remaining $n-m$ columns to obtain a set of $n$ linearly independent columns.  We prove this inductively. \\~ \\
\textbf{When $\mathbf{m=1}$}. We adapt the arguments given in Ref. \cite{user0}\footnote{The proof in Ref. \cite{user0} is itself taken from Ref. \cite{Newman} (Corollary II.I). We can't rely on the proofs in Ref. \cite{user0,Newman} directly, since these results are valid only for commutative rings with no zero divisors. In our case, when $d$ is non-prime, $\Z$ has zero divisors.} to our purpose. To say the $\vect{a_1}=\vect{a}$ is linearly independent means that there is no non-zero $x \in \Z$ such that $x \vect{a}=0$. In other words, $\vect{a}$ is of order $d$. This implies that if $\vect{a}=(a_1,a_2,\cdots,a_n)^T$, then as integers in $\mathbb{Z}$, $\mathrm{g.c.d.}\left\{a_1,a_2,\cdots,a_n,d\right\}=1 \in \mathbb{Z}$. We can thus apply B\'ezout's identity (see Ref. \cite{Bezout}), in the same way that it is applied in Ref. \cite{user0}. For completeness, we give the whole proof here.  \newline 
Proceeding inductively, suppose that $n=2$ first. B\'ezout's identity tells us that there exist $b_1$, $b_2$ and $d'$ in $\mathbb{Z}$ such that $a_1 b_2 - a_2 b_1 + d' d = 1$ in $\mathbb{Z}$. Thus we also get that $\left(a_1 b_2 - a_2 b_1 \right) \mod d = 1 \in \Z$. With $A=\SmallMatrix{a_1 & b_1 \\ a_2 & b_2 } \mod d$, the result is proved for $n=2$. Next, suppose that the result is true for $n-1$. We need to then prove the result for $n$. As before $\mathrm{g.c.d}\left\{ a_1, a_2, \cdots, a_n, d \right\}=1 \in \mathbb{Z}$. There exists some integer $g$ which is a common factor of $a_1$, $a_2$, $\cdots$ and $a_{n-1}$ such that if $b_i = a_i/g$, then $\mathrm{g.c.d}\left\{b_1, b_2, \cdots, b_{n-1},d \right\}=1$. It isn't necessary that $g$ equals $g'\coloneqq \mathrm{g.c.d}\left\{a_1, a_2, \cdots, a_{n-1},d \right\}$, because $\mathrm{g.c.d}\left\{ a_1/g', a_2/g', \cdots, a_{n-1}/g', d\right\}$ need not be $1$. Writing $d=g'h'$, we see that $a_i/g'$ have no common factors with $h'$ (otherwise $g'$ would have been larger). Thus the only factors which $a_1/g'$, $a_2/g'$, $\cdots$, $a_{n-1}/g'$ have in common with $d$ will occur in $g'$. Writing $\mathrm{g.c.d}\left\{ a_1/g',a_2/g',\cdots,a_{n-1}/g',g'\right\}\eqqcolon\tilde{g}'$, and defining $g''=\tilde{g'}g'$, we examine whether $a_1/g''$, $a_2/g''$, $\cdots$, $a_{n-1}/g''$ and $d$ have any common factors. If not, we may choose $g=g''$, otherwise we iterate recursively. These iterations have to stop at some point since the numbers $a_1$, $a_2$, $\cdots$, $a_{n-1}$ are finite. Note that $g/g'$ is necessarily a factor of $g'$ itself. Since the hypothesis is assumed true for $n-1$, corresponding to $\vect{b}\in \mathbb{Z}^{n-1}$, where $\vect{b}=(b_1,b_2, \cdots, b_{n-1})^T$, we can construct an $(n-1)\times(n-1)$ matrix $B$ whose first column is $\vect{b}$ and such that $\det B + d'd =1 \in \mathbb{Z}$. Construct an $n \times n$ matrix $A$ in the same fashion as in Ref. \cite{user0}. 
\begin{equation}
\label{eq:Auser0}
A^T = \begin{pmatrix} gb_1 & gb_2 \cdots & gb_{n-1} & a_{n1}  \\ 
 b_{21} & b_{22} \cdots & b_{2,n-1} & 0  \\ 
 \vdots & \vdots \ddots & \vdots & \vdots \\ 
 b_{n-1,1} & b_{n-1,2} \cdots & b_{n-1,n-1} & 0  \\
 rb_1 & rb_2 \cdots & rb_{n-1} & s  \end{pmatrix},
\end{equation} which gives us 
$\det A = \left( s g -  a_{n1}   r  \right) \det B $, and since $\det B = 1 -d'd$, we get $\det A = \left( s g - a_{n1} r \right)\left( 1 - d'd \right)$. Note that $\mathrm{g.c.d} \left\{ g', a_{n1},d \right\}=1$, which tells us that $g'$ and $a_{n1}$ are co-prime. Since $g/g'$ is also a factor of $g'$, we get that $\mathrm{g.c.d} \left\{ g, a_{n1},d \right\}=1$ as well. Hence B\'ezout's identity informs us that there will exist some integers $s$, $r$ and $d''$ in $\mathbb{Z}$ so that $s g - a_{n1} r + d'' d= 1 $. Hence we get that $\det A = (1-d'd)(1-d''d)$. In $\Z$, $\det A = 1$, which implies that $A \mod d$ is invertible. This proves our result for $m=1$.
\\ ~ \\ 
\textbf{Assumed true for $m$. To prove for $m+1$.}  $\vect{a}_1$, $\vect{a}_2$, $\cdots$, $\vect{a}_m$, $\vect{a}_{m+1} \in \Zn$, such that they are LI. Our hypothesis is assumed true for $m$ vectors, so $\vect{a}_1$, $\vect{a}_2$, $\cdots$, $\vect{a}_m$ can be placed as the (left-most) columns of an $n \times n$ matrix $A$, which is invertible. Consider the isomorphism on $\Zn$: $\vect{a} \rightarrow A^{-1} \vect{a}$. This maps $\vect{a_j}=\vect{e}_j$ for $j=1$ to $m$, where $\vect{e}_j$ are the standard basis vectors whose only non-zero component is the $j$-th component, and this component is equal to $1$. Let $A^{-1} \vect{a}_{m+1} = (c_1,c_2,\cdots,c_m,c_{m+1},c_{m+2},\cdots,c_n)^T \eqqcolon \vect{c}$. Since $A^{-1}$ is an isomorphism on $\Zn$, we get that $\vect{e}_1$, $\vect{e}_2$, $\cdots$, $\vect{e}_m$ and $\vect{c}$ are LI. Consider the following.
\begin{equation}
    \label{eq:matrix_multi}
    \begin{pmatrix}
    \unity_m &  0 \\ 
        0 &  A'
    \end{pmatrix} \begin{pmatrix} \begin{matrix} \unity_m \\ 0 \end{matrix} & \rvline & \begin{matrix} \vect{c} \end{matrix} 
    \end{pmatrix},
\end{equation}
where the matrix on the left is an $n \times n$ matrix. The upper right $0$ denotes a block matrix of size $m \times n-m$ whose entries are all $0$. The lower left $0$ denotes another block matrix of size $(n-m) \times m$ and whose entries are all $0$. $A'$ is an $(n-m) \times (n-m)$ matrix, as yet unspecified. The matrix on the right is of size $n \times (m+1)$, whose first $m$ columns are the vectors $\vect{e}_j$ ordered from $1$ to $m$, and the $(m+1)$-th column is $\vect{c}$. Denote $\vect{c}' \coloneqq (c_{m+1}, c_{m+2},\cdots, c_n)^T$. Our hypothesis allows us to choose $A'$ to be invertible and such that $A' \vect{c}'= (1,0,\cdots, 0)^T$, which is a column vector of size $n-m$. Defining $\vect{c}''\coloneqq (c_1, c_2, \cdots, c_m)^T$, so that $\vect{c}^T=(\vect{c}''^T,\vect{c}'^T)$, define an $n \times n$ matrix $T$
\begin{equation}
    \label{eq:matrix_multi_T}
    T = \begin{pmatrix}
    \unity_k & \rvline & -\vect{c}'' & \rvline & 0 \\
    0 & \rvline & 1 & \rvline & 0 \\
    0 & \rvline & 0 & \rvline & \unity_{n-k-1} 
    \end{pmatrix}.
\end{equation} The diagonal blocks $\unity_k$, $1$ and $\unity_{n-k-1}$ are of sizes $k$, $1$ and $n-k-1$ respectively. These determine the sizes of the remaining blocks of $T$. When $T$ is multiplied on the left of Eq. \eqref{eq:matrix_multi}, we get an $n \times (m+1)$ sized matrix whose columns are $\vect{e}_1$, $\vect{e}_2$ $\cdots$, $\vect{e}_m$ and $\vect{e}_{m+1}$. Denoting the left matrix in Eq. \eqref{eq:matrix_multi} by $A''$, and defining $Q \coloneqq T A'' A^{-1}$, we see that $Q$ is an invertible matrix, and the action of $Q$ on $\vect{a}_j$ is to produce $\vect{e}_j$ for $j=1$ to $m+1$. This tells us that $\vect{a}_j$ is the $j$-th column of the matrix $Q^{-1}$. Thus given $\vect{a}_j$ for $j=1$ to $m+1$, we may extend it to a set of $n$ LI vectors by adding the remaining columns as vectors to this set. This proves our result. 
\section{Extension of a check matrix to a symplectic matrix}
\label{appendix:symplectic_realization} \tanmayc{\begin{rem} For the case when $d$ is a prime (or more generally, when $\Z$ is a field), this section reduces to a proof of Witt's extension theorem. A proof of Witt's extension theorem for the general case is hard to find. Thus we include a proof here for the interested reader.\end{rem}} 
Consider the stabilizer group $S$ for an $[[n,k]]_d$ stabilizer code. For a choice of independent generators $g_j$ of $S$ and an ordering, one may readily construct the check matrix $H$ using these generators (see Eq. \eqref{eq:checkmatrix}). The columns of $H$ are linearly independent and satisfy the commutation relations given by Eq. \eqref{eq:condition1}. Our goal in this section of the Appendix is to prove that the $2n \times (n-k)$ matrix $H$ may be extended to a $2n \times 2n$ symplectic matrix \tanmayc{$S$}.
\begin{prop}{Proposition 10.4, \cite{Nielsen}}
For each $i=1,2\cdots,n-k$, there exists some $g \in \Pn$, such that $g_i g = \omega g g_i $, and $g_jg=gg_j$ when $j\neq i$. 
\end{prop}
\begin{proof}
Our proof will be based on the proof in \cite{Nielsen}, which works perfectly well for the case when $d=p$ is a prime, and $\Z =\F$ is a field. In that case, one simply seeks the solution to the following equation: $ H^T \L \vect{x} = \vect{e_i}$, where $\vect{x}\in\F^{2n}$ is a column vector whose solution is sought, and $\vect{e_i}\in \F^{n-k}$, is the vector whose only non-zero element is the $i$-th component, and this component is equal to $1$. The existence of the solution the aforementioned equation is based on the fact that $\dim \F^{n-k}=\mathrm{rank}{H^T \L}$, thus the transformation $\vect{x} \rightarrow H^T \L \vect{x} $ is surjective in $\F^{n-k}$. To establish that the same result holds true when $d$ is a composite number, we use the following arguments: since $\vect{a_1}$, $\vect{a_2}$, $\cdots$, $\vect{a_{n-k}}$ are LI, one may extend this to a basis of $2n$ columns vectors in $\Znn$ (see Subsection \ref{subsec:app3}). Construct a $2n \times 2n$ matrix $\He$, by adding to $H^T$ the additional $n+k$ basis vectors, by transposing them and placing the corresponding rows below $H^T$. Since the rows of $\He$ are linearly independent and since they are $2n$ in number, by arguments in Subsection \ref{subsec:app1} they form a basis. Hence solutions to the equations $\vect{x}^T \He  = \vect{e_j}^T$ may be found for all $j=1,2,\cdots, 2n$, which implies that $\He$ has a left-inverse. Since the left-inverse of $\He$ is also its right inverse, that implies that the columns of $\He$ are also linearly independent, and since these columns are $2n$ in number, they also form a basis using the arguments in Subsec. \ref{subsec:app1}. Thus one may then find the solution to the equation:
\begin{equation}
    \label{eq:He}
    \He \ \vect{x'} \ = \ \vect{e'_i},
\end{equation}
where $\vect{x'}, \vect{e'_i}  \in \Znn$. Here $\vect{e'_i}$ is the vector whose only non-zero entry is the $i$-th entry, and this entry is equal to $1$, whereas $\vect{x'}$ is a vector whose solution is sought. That $\He$ is invertible tells us that $\vect{x'} = {\He^{-1}} \vect{e'_i}$. Define $\vect{x} \coloneqq \L^{-1}\vect{x'}$. Then $\vect{x}$ satisfies the equation: $\He \L \vect{x} = \vect{e'_i}$. One may neglect the $n+k$ bottom most rows of $\He$ and of $\vect{e'_i}$ (not of $\L$, or of $\vect{x}$ - their rows must be maintained intact), which then gives us the equation $H^T \L \vect{x} = \vect{e_i}$. Let $g$ be the Pauli element whose $\Znn$ representative is $\vect{x}$. The commutation relations between $g$ and the $g_j$'s is then determined by the equation $H^T \L \vect{x} = \vect{e_i}$: we get that $g$ and $g_j$ commute when $j\neq i$ and $g_ig=\omega g g_i$. Hence proved.
\end{proof}
Let us obtain the solutions for $\vect{x}$ for all $i=1,2,\cdots,n-k$, and call the corresponding solutions $\vect{w'_i}$. We want $\vect{w'_i}^T\L\vect{w_j}=0$ for all $i,j=1,2,\cdots, n-k$. To that end, we perform a Gram-Schmidth orthogonalization procedure in the following way, starting with $\vect{w_1}\coloneqq \vect{w'_1}$, and starting with $i=2$ to $i=n-k$,
\begin{align}
    \label{eq:gm_schmidt}
    & \vect{w_i} \ \coloneqq \ \vect{w'_i} \ - \  \sum_{j=1}^{i-1} \left( \vect{w_j}^T \L \vect{w'_i}  \right) \vect{a_j}.
\end{align}
If $k=0$, our job is done: to extend $H$ to a $2n \times 2n$ symplectic matrix, we need to add the columns $\vect{w_i}$ to the left (or right) of $H$, and this will give us our symplectic matrix $\tanmayc{S}$. Suppose $k=1$, we need to find pairs of vectors $\vect{a_{n}}, \vect{w_{n}}$ so that 
\begin{align}
    \label{eq:extension1}
    & \vect{a_{n}}^T \Lambda  \vect{a_j} \ = \ 0, \ \forall  \ j=1,2,\cdots, n, \notag \\
    & \vect{w_{n}}^T \Lambda  \vect{w_j} \ = \ 0, \ \forall  \ j=1,2,\cdots, n, \notag \\
    & \vect{w_{n}}^T \Lambda  \vect{a_j} \ = \ \delta_{j,n}, \ \forall  \ j=1,2,\cdots, n.
\end{align}
At the moment, we have $2(n-1)$ linearly independent vectors $\vect{a_j}$ and $\vect{w_j}$. Let $\vect{a'}$ be some other vector so that the $\vect{a'}$, $\vect{a_j}$ and $\vect{w_j}$ form a set of $2n-1$ LI vectors in $\Znn$. Define $\vect{a_n} \coloneqq \vect{a'} -  \  \sum_{j=1}^{n-1} \left( \vect{w_j}^T \L \vect{a'}  \right) \vect{a_j}-  \  \sum_{j=1}^{n-1} \left( \vect{a_j}^T \L \vect{a'}  \right) \vect{w_j}$. Since the linear independence of $\vect{a'}$, $\vect{a_j}$ and $\vect{w_j}$ ensures that $\vect{a_n}$ is non-zero. We have also ensured that $\vect{a_n}$ satisfies the conditions $\vect{a_j}^T \L \vect{a_n} = \vect{w_j}^T \L \vect{a_n}=0$ for all $j=1,2,\cdots,n-1$. It remains to find $\vect{w_n}$ with the properties as desired by Eq. \eqref{eq:extension1}. Such a $\vect{w_n}$ may be found by using Proposition 10.3 again for an extended check matrix which is obtained by adding $\vect{a_n}$ to the original check matrix $H$. Thus we have obtained all the column vectors as desired, and the newly obtained column vectors may be arranged suitably to produce a $2n \times 2n$ symplectic matrix $\tanmayc{S}$. \newline 
If $k>1$, we may apply the same procedure as above iteratively, to obtain the symplectic matrix $\tanmayc{S}$.
\section{Symplectic conditions for $M$}
\label{appendix:M_sycondition}
In this section of the Appendix, we derive the symplectic conditions Eq. (\ref{sycondition1})-(\ref{sycondition4}) for the matrix $M$ in Eq. \eqref{eq:M_matrix}, which has the form:
\begin{equation}
    \label{eq:M_matrix2}
M = 
\left[\begin{matrix}
M_{11}&M_{12}&0&M_{14}\\
M_{21}&M_{22}&0&M_{24}\\
M_{31}&M_{32}&A&M_{34}\\
M_{41}&M_{42}&0&M_{44}
\end{matrix}\right].
\end{equation}
Here $M_{11}$, $M_{31}$, $M_{13}$ (which is $0$) and $A$ are of dimensions $(n-k)\times(n-k)$, $M_{22}$, $M_{24}$, $M_{42}$ and $M_{44}$ are of dimensions $2k \times 2k$, $M_{21}$, $M_{41}$, $M_{23}$ and $M_{43}$ (which are both $0$) are of dimensions $k \times (n-k)$, and finally, $M_{12}$, $M_{14}$, $M_{32}$ and $M_{34}$ are of dimensions $(n-k)\times k$.  
$M$ satisfies the following symplectic condition.
\begin{equation}
\label{eq:sympcondition}
M^T\L M = \L.
\end{equation}
It will be convenient to label the LHS of Eq. \eqref{eq:sympcondition} as follows.
\begin{equation}\label{eq:Msymp_blocks}
M^T \L M \ = \ 
\left[\begin{matrix}
Q_{11}&Q_{12}&Q_{13}&Q_{14}\\
Q_{21}&Q_{22}&Q_{23}&Q_{24}\\
Q_{31}&Q_{32}&Q_{33}&Q_{34}\\
Q_{41}&Q_{42}&Q_{43}&Q_{44}\\
\end{matrix}\right].
\end{equation}
Here we only consider blocks $Q_{13}, Q_{23}$ and $Q_{43}$. That is the conditions given in Eq. (\ref{sycondition0_1}) and (\ref{sycondition0_2}):
\begin{equation}
\label{cond1}
\begin{aligned}
& Q_{13}=M^T_{11}A=I_{k} \Longrightarrow M_{11}=(A^{-1})^T, \\
& Q_{23}=M^T_{12}A=0 \Longrightarrow M_{12}=0, \\
& Q_{43}=M^T_{14}A=0 \Longrightarrow M_{14}=0. \\
\end{aligned}
\end{equation}
\tanmayc{Substituting} Eq. (\ref{cond1}) into Eq. (\ref{eq:M_matrix}) gives us Eq. \eqref{eq:M_form}. Now, Eq. \eqref{sycondition2}, Eq.\eqref{sycondition3} and Eq. \eqref{sycondition4} are merely the equations for the blocks $Q_{11}$, $Q_{12}$ and $Q_{14}$ respectively, whereas Eq. \eqref{sycondition1} is the collective condition on $\begin{psmallmatrix} Q_{22} & Q_{24} \\ Q_{42} & Q_{44} \end{psmallmatrix}$. 
\section{Significance of the subgroups of $\Tnkd$: $\Spnd{2k}{d}$, $\GLnd{n-k}$, $\BA$ and $\BS$}
\label{app:sec:significance}
We now explain the significance of each of these subgroups. For an $[[n,k]]_d$ stabilizer code, let $H$ be the corresponding check matrix and let a symplectic extension of $H$ be given by $M$ as in Eq.\eqref{eq:M_matrix}. One may obtain another symplectic matrix of the same stabilizer code by right-multiplying $M$ with a matrix in $\Tnkd$ with a decomposition $M_EM_SM_L$. Suppose that $M_E=M_L=I_{2n}$. Right-multiplying $M$ with $M_S$ transforms the check matrix $H$ to another check-matrix $HA$ (of the same code). Thus the $\GLnd{n-k}$ subgroup realises the freedom of making different choices of the independent generators of said stabilizer group. Right-multiplying $M$ with $M_L$ transforms the columns in $L_X$ and $L_Z$, while leaving $E$ and $H$ invariant. Thus the $\Spnd{n-k}{d}$ subgroup realises the freedom of  making different choices of the generators of the {\it purely} logical Pauli group, which we denote by $\Pl$. Each element of this subgroup (except the identity) will act non-trivially on the stabilizer code, i.e., codewords are mapped to other codewords under the action of these elements. Thus the {\it purely} logical Pauli group is a subgroup of the logical Pauli group $N(S)$, which contains the stabilizer group $S$. One anticipates that the choice of generators for $S$ (which is realised by the $\GLnd{n-k}$ subgroup) is independent of the choice of generators of the {\it purely} logical Pauli group (which is realised by the $\Spnd{k}{d}$ subgroup). This agrees with the fact that both subgroups commute with each other. \newline We next explain the significance of $M_E$. Let $g_E$ be a Pauli which anti-commutes with some generators of $S$. Consider $g_E \Pl$ which is a left coset of $\Pl$. Each element $g_Eg_L$ rotates $\CS$ to the same orthogonal coding space $\CS^{\perp}$. But for different choices of $g_L$, $g_E g_L$ performs a distinct internal rotation {\it within} $\CS^{\perp}$. The error correction procedure is unable to distinguish between different internal rotations due to different $g_L$. Thus it can correct only one element in the coset $g_E \Pl$. This is true for all $g_E$'s which anticommute with elements in $S$. Which element within $g_E \Pl$ is {\it chosen} to be the correctable error depends on the error correction model. We next prove the following.
\begin{thm}
\label{thm:E}
The columns $E$ in $M$ in Eq. \eqref{eq:M_matrix} represent correctable errors. Right multiplying $M$ with different $M_E \in \BA$ realises the freedom in making different choices of such correctable errors. The action of $M_T \in \BS$ doesn't affect this choice. 
\end{thm}
\begin{proof}
That the columns of $E$ fix the choice of the correctable errors of the stabilizer code is well-known in the stabilizer formalism for quantum error correction (see Chapter 10 in \cite{Nielsen} for instance or see the beginning in Section III in \cite{Poulain}). What we need to prove is that for each distinct choice of correctable errors, there is a choice of $M_E \in \BA$ such that right multiplying $M$ with $M_E$ produces the corresponding column block of correctable errors in $M M_E$. To this end, suppose that the columns in $E$ are homomorphic images of $g_1$, $g_2$ $\cdots$, $g_k \in \Pn$ (see Eq. \eqref{eq:arbitrary_n_qudit_pauli_group}). Suppose that we wish for $g_1h_1$, $g_2h_2$, $\cdots$, $g_kh_k$ to be correctable errors instead, where $h_1, h_2 \cdots, h_k \in \Pl$. These errors may be represented in $\Znn$ by the $2n \times k$ matrix 
\begin{equation}
    \label{eq:Eprime}
    E' \ = \ E + L_XN' + L_ZL', 
\end{equation}
where the $j$-th columns of $N'$ and $L'$ are determined by how $h_j$ decomposes into the generators of $\Pl$. Thus $N'$ and $L'$ are determined entirely by $h_j$'s. Equation \eqref{eq:Eprime} may be realised by right multiplying $M$ with a matrix of the form $M_E$ with $N'=N$ and $L'=L$, and $K_A$ determined entirely by $N$ and $L$ as follows: $K_A = L^TN-N^TL$. This is just Eq. \eqref{sycondition2}. Additionally, right multiplying $M$ by any $M_T \in \BS$ changes the columns in $E$ to $E' = E + HK_S$. If $g_{E'}$ and $g_E$ represent the $j$-th columns in $E'$ and $E$, then $g_{E'}=g_Eg_S$, where $g_S \in S$ is some element of the stabilizer group. $g_S$ is determined by the $j$-th column of $K_S$ and the choice of generators in $H$. Note that the error operators $g_{E'}$ and $g_E$ are both simultaneously correctable since the stabilizer element $g_S$ acts trivially on the code. Hence proved.
\end{proof}
\subsection{Order of $\Spnd{2n}{p^m}$}
\label{app:sec_sp_order}
A computation of the order of $\Spnd{2n}{p^m}$ is given in Ref. \cite{hattice}. For the interested reader, we reproduce that computation here, while elaborating on some of the definitions and computations to make them clearer.
 Define a ring homomorphism $\psi_1: \ \Zj{m} \rightarrow \Zj{m-1}$ as
\begin{equation}
\label{eq:psi1}
\psi_1(x) = x \: \: \text{mod} \: \: p^{m-1},
\end{equation}
where $x \in \Zj{m}$. The following remark explains Eq. \eqref{eq:psi1}. \begin{rem}
\label{rem:exp}
The ideal generated by $p^{m-1}$ within $\Zj{m}$ is isomorphic to $\Zj{}$. We label this ideal as $p^{m-1} \Zj{m}$. The factor ring of $\Zj{m}/ p^{m-1}\Zj{m}$ is isomorphic to $\Zj{m-1}$. $\psi_1$ is the map from $\Zj{m} \rightarrow \Zj{m}/p^{m-1}\Zj{m}$. $\psi_1$ is a ring homomorphism, and $\ker \psi_1 = p^{m-1} \Zj{m}$. 
\end{rem} We define a group homomorphism from $\Spnd{2n}{p^m} \rightarrow \Spnd{2n}{p^{m-1}}$ as follows. Let $M$ be a $2n \times 2n$ symplectic matrix.  Define $\Xi_1: \Spnd{2n}{p^m} \rightarrow \Spnd{2m}{p^{m-1}}$ as follows: the $ij$-th matrix element of $\Xi_1(M)$ is $\psi_1 \left( M_{ij} \right)$. Since $M \in \Spnd{2n}{p^m}$, it satisfies the equation $M^T \L M = \L$ (where $\L \in \Spnd{2n}{p^m}$), then \begin{equation}
    \label{eq:psi_symplectic}
   \left(  \Xi_1(M) \right)^T \ \L \ \Xi_1(M) \ = \ \L, \ \mathrm{where} \ \L \in \Spnd{2n}{p^{m-1}},
\end{equation} which implies that $\Xi_1(M) \in \Spnd{2n}{p^{m-1}}$. That $\psi_1$ is a ring homomorphism immediately implies that $\Xi_1$ is group homomorphism from $\Spnd{2n}{p^m}$ to $\Spnd{2n}{p^{m-1}}$. $\ker \Xi_1$ is the subgroup in $\Spnd{2n}{p^m}$ which $\Xi_1$ maps to  $I_{2n} \in \Spnd{2n}{p^{m-1}}$. Note that $I_{2n}$ (in $ \Spnd{2n}{p^m}$) lies in $\ker \Xi_1$. Thus any $M \in \ker \Xi_1$ can be written in the form $I_{2n}+ K$, where $K$ is a $2n\times 2n $ matrix with matrix elements in $p^{m-1}\Zj{m}$, and satisfies the symplectic equation
\begin{align}
\label{eq:symplectic_K}
 & \left( I_{2n} + K \right)^T \ \L \ \left( I_{2n} + K \right) \ = \ \L, \notag \\ 
 \Rightarrow  \ & \ K^T \L \ + \ \L K \ + \ K^T \L K \ = 0.
\end{align}
Note that $K^T \L K = 0$ since $K$ can be written as $K= p^{m-1} K'$, with $K' \in \left\{0,1,\cdots, p-1 \right\} \subset \Zj{m}$, which gives us $K^T \L K = p^m \ p^{m-2} K'^T \L K'$ which is $0$ since $p^m =0$ in $\Zj{m}$. Thus $K$ has to satisfy the equation $\L K = -K^T \L$. Decomposing $K$ into $n \times n$ blocks as follows:
\begin{equation}
    \label{eq:K_block}
K = 
\begin{bmatrix}
X&Y\\
Z&W
\end{bmatrix},
\end{equation} we see that the symplectic condition becomes $X=-W^T$, $Y=Y^T$ and $Z=Z^T$. Thus $K$ has $2n^2 + n$ unconstrained matrix elements which take values in $p^{m-1} \Zj{m}$, which has $p$ elements. Thus $\left| \ \ker \Xi_1 \  \right| \ = \ p^{2n^2+n}$ and we get \begin{equation}\label{eq:order_Sp_1} 
\left| \ \Spnd{2n}{p^m} \  \right| \ = \ p^{2n^2+n} \   \left| \ \Spnd{2n}{p^{m-1}} \  \right|.\end{equation} Similarly one may define $\psi_j:\Zj{m-j+1} \rightarrow \Zj{m-j}$ and $\Xi_j:\Spnd{2n}{p^{m-j+1}} \rightarrow \Spnd{2n}{p^{m-j}}$, with $\left| \ \ker \Xi_2 \  \right| \ = \ p^{2n^2+n}$ for $j=2,3,\cdots, m-1$, and \begin{equation}\label{eq:order_Sp_j}
\left| \ \Spnd{2n}{p^{m-j+1}} \  \right| \ = \ p^{2n^2+n} \   \left| \ \Spnd{2n}{p^{m-j}} \  \right|,\end{equation} which finally gives us
\begin{equation}\label{eq:order_Sp}
\left| \ \Spnd{2n}{p^{m}} \  \right| \ = \ p^{(m-1)(2n^2+n)} \   \left| \ \Spnd{2n}{p} \  \right|.\end{equation} It is known (see Ref. \cite{Grove}) that $\left| \Spnd{2n}{p} \right| \ = \ p^{n^2} \prod_{j=1}^{n} \left(p^{2j}-1 \right)$, which tells us that 
\begin{equation}
    \label{eq:order_sp}
    \left| \ \Spnd{2n}{p^m} \  \right|  \ = \ p^{(2m-1)n^2+(m-1)n} \ \prod_{j=1}^{n} \left( p^{2j} -1 \right).
\end{equation}
\tanmayc{
\section{Proof that the Chinese remainder theorem implies  Eq. \eqref{eq:Chinese_remainder_theorem_cor}}
\label{app:chinese_remainder_thm}
\subsection{What is the Chinese remainder theorem?}
\label{subsec:CRT_explanation} For a reader, who may be unfamiliar with the Chinese remainder theorem (CRT), we briefly explain it in this subsection. \newline ~ \newline 
We use the following abbreviation for $\mathbb{Z}_{p_1}^{m_1} \times \mathbb{Z}_{p_2}^{m_2} \times \cdots \times \mathbb{Z}_{p_r}^{m_r}$. \begin{align}\label{eq:CRT_K} & K \equiv \mathbb{Z}_{p_1}^{m_1} \times \mathbb{Z}_{p_2}^{m_2} \times \cdots \times \mathbb{Z}_{p_r}^{m_r}. \end{align} Since $K$ is a direct product of rings $\mathbb{Z}_{p_i^{m_i}}$, addition and multiplication are defined component-wise. Let $x_i,y_i \in \mathbb{Z}_{p_i^{m_i}}$ for all $i=1,2,\cdots,r$.
\begin{align}
\label{eq:K_arithmetic_plus}
& (x_1, x_2, \cdots, x_r) + (y_1, y_2, \cdots, y_r) \ = \ (x_1 + y_1, x_2 + y_2, \cdots, x_r + y_r),\ \mathrm{and} \\ 
\label{eq:K_arithmetic_times}
& (x_1, x_2, \cdots, x_r) \ \  . \  (y_1, y_2, \cdots, y_r) \ = \ (x_1 . y_1, \  x_2 . y_2, \  \cdots, \  x_r . y_r \ ). 
\end{align}
CRT, as expressed by Eq. \eqref{eq:Chinese_remainder_theorem}, establishes a ring isomorphism, which we call $\Gamma$. $\Gamma(x)$ for any $x \in \Z$ is computed as follows.
\begin{align}
\label{eq:gamma}
& \Gamma: \Z \rightarrow  K, \notag \\
& \Gamma (x) \ = \ \left( x \mod p_1^{m_1}, \ x \mod p_2^{m_2}, \cdots , x \mod p_r^{m_r} \right), 
\end{align} which we abbreviate as $\Gamma(x) = \left( x_i\right)_i$, where $x_i \equiv x \mod p_i^{m_i}$. CRT states the following.
\begin{enumerate}
\item $\Gamma$, as defined in Eq. \eqref{eq:gamma}, is a bijection from $\Z$ to $K$. 
\item Arithmetic of pre-images (in $\Z$) is isomorphic under $\Gamma$, to arithmetic of the images (in $K$). That is, for $x,y \in \Z$, and $\Gamma(x)=(x_i)_i$ and $\Gamma(y)=(y_i)_i$, we have
\begin{align}
\label{eq:CRT_add}
&\Gamma(x+y) \ = \ (x_i + y_i)_i, \ \mathrm{and} \\
\label{eq:CRT_multiply}
&\Gamma(x.y) \ = \ (x_i .y_i)_i,
\end{align} where $x_i + y_i$ and $x_i y_i$ are the addition and multiplication operations, carried out in $\mathbb{Z}_{p_i^{m_i}}$.
\end{enumerate} 
\subsection{Extending CRT to matrices}
\label{subsec:CRT_matrix} Next, consider the following direct product of matrix rings $\Mtn{i}$ for $i=1$ to $i=r$. \begin{equation}
\label{eq:CRT_L}  L \equiv  \Mtn{1} \times \Mtn{2} \times \cdots \times \Mtn{r} .\end{equation} Addition and multiplication in $L$ are performed component-wise. \begin{align}\label{eq:arithmetic_L_plus} & (S_1, S_2, \cdots, S_r) + (T_1, T_2, \cdots, T_r) \ = \ (S_1 +T_1, S_2 + T_2, \cdots, S_r+T_r), \ \mathrm{and} \\ \label{eq:arithmetic_L_times}   & (S_1, S_2, \cdots, S_r) . (T_1, T_2, \cdots, T_r) \ = \ (S_1.T_1, S_2 . T_2, \cdots, S_r.T_r),  \end{align} where $S_i, T_i \in \Mtn{i}$. \newline ~ \newline We may extend CRT to derive a ring isomorphism $\GM: \M \rightarrow L$ as follows. Let $S\in\M$. Then
\begin{equation}
\label{eq:def_GM}
 \GM(S) \ = \  \left( S \mod p_1^{m_1},S \mod p_2^{m_2},\cdots,S \mod p_r^{m_r} \right),
\end{equation} where by $S \mod p_i^{m_i}$ we mean that matrix in $\Mtn{i}$ whose $jk$-th matrix element is $S_{jk} \mod p_i^{m_i}$. We will use the abbreviation $\left( S \mod p_1^{m_1},S \mod p_2^{m_2},\cdots,S \mod p_r^{m_r} \right) = (S \mod p_i^{m_i})_i$ or simply as $(S_i)_i$ whenever convenient. That $\GM$ is a ring isomorphism, is seen from the following.
\begin{enumerate}
\item $\GM$ is a bijection because, firstly, $\left| \M \right| \ = \ \Pi_{i=1}^r \left( p_i^{m_i}  \right)^{2n} \  =  \  \left| L \right| $, since $d= \Pi_{i=1}^{r} p_i^{m_i}$. Secondly, every element in $L$ has a pre-image in $\M$. 
\item Matrix addition and multiplication in $\M$ is homomorphic, under $\GM$, to addition and multiplication defined in Eq. \eqref{eq:arithmetic_L_plus} and Eq. \eqref{eq:arithmetic_L_times} respectively. Let $S,T \in \M$. For addition: 
\begin{equation}
\label{eq:GM_plus}
\GM(S+T) = \left( ( S+T) \mod p_i^{m_i} \right)_i \ = \ \left (  S \mod p_i^{m_i}  \ + \  T \mod p_i^{m_i}  \  \right)_i \ = \ \GM(S) + \GM(T),
\end{equation} where we used Eq. \eqref{eq:arithmetic_L_plus} to obtain the right-most expression from the previous expression. For multiplication:
\begin{equation}
\label{eq:GM_times}
\GM(S.T) = \left( ( S.T) \mod p_i^{m_i} \right)_i \ = \ \left ( \left( S \mod p_i^{m_i} \right) \ . \ \left( T \mod p_i^{m_i} \right) \  \right)_i \ = \ \GM(S) . \GM(T),
\end{equation} where we used Eq. \eqref{eq:arithmetic_L_times} to obtain the right-most expression from the previous expression.
\end{enumerate} 
\subsection{Restricting $\GM$ to $\Sp$}
\label{subsec:GM_Sp}
Consider the following scenario. 
\begin{itemize}
\item We restrict $\GM$ to acting on elements in $\Sp$. 
\item Additionally, we restrict matrix operations to multiplication \emph{only}. We dispense with addition.
\end{itemize}
$\Sp$ is a group. That $\GM$ is a bijection from $\Sp$ to $\GM(\Sp)$, and that $\GM$ satisfies Eq. \eqref{eq:GM_times} implies that $\GM(\Sp)$ must also be a group, and moreover that $\Sp$ and $\GM(\Sp)$ are isomorphic. We now prove that for any $S \in \Sp$, $\GM(S) \in G$, where 
\begin{equation}
\label{eq:direct_prod_sp} 
G \ \equiv \  \Spnd{2n}{p_1^{m_1}} \times \Spnd{2n}{p_2^{m_2}} \times \cdots \times  \Spnd{2n}{p_r^{m_r}}.
\end{equation} 
Since $S \in \Sp$ it satisfies the symplectic equation: $S^T \L S \ = \ \L$. Apply $\GM$ to both sides of the symplectic equation: $\GM\left(S^T \L S \right)  \ = \ \GM(\L)$. Since $\GM (\L) \ = \ (\L_i)_i$, we get that  
\begin{equation}
\label{eq:symp_mod}
\left( \left( S^T \L S \right)_i \right)_i \   =  \ \left(  S_i^T \L_i S_i \right)_i = \left( \L_i \right)_i. \end{equation} Note that $\L_i$ takes the same form as $\L$ from Eq. \eqref{eq:Lambda}, but the $1$'s and $0$'s in $\L_i$ play the role of the multiplicative and additive identities in $\mathbb{Z}_{p_i^{m_i}}$. Eq. \eqref{eq:symp_mod} implies that $S_i \in \Spnd{2n}{p_i^{m_i}}$, which proves that $\GM (S) \in G$. This proves the first isomorphism in Eq. \eqref{eq:Chinese_remainder_theorem_cor}. 
\subsection{Restricting $\GM$ to $\Tnkd$} 
\label{subsec:GM_tnkd} Further restricting the action of $\GM$ to only $\Tnkd$, similar arguments as above inform us that $\GM(\Tnkd)$ is a group, and moreover, $\GM(\Tnkd)$ is isomorphic to $\Tnkd$. We prove now that for any $S \in \Tnkd$, $\GM(S) \in H$, where 
\begin{equation}
\label{eq:GM_H}
H \ \coloneqq \  \Tnk{p_1^{m_1}} \times \Tnk{p_2^{m_2}} \times \ldots   \times \Tnk{p_r^{m_r}}. \end{equation} That $S \in \Tnkd$ entails two things (see Thm.~\ref{tnk}): (i) $S \in \Sp$ and (ii) $S$ has the form of Eq. \eqref{eq:M_matrix_block}, i.e., the matrix blocks $S_{13}=0_{n-k}$, $S_{23}=S_{43}=0_{k\times (n-k)}$. From Subsection \ref{subsec:GM_Sp} we learn that when $S \in \Sp$, then $\GM(S) \in G$. And from Eq. \eqref{eq:def_GM} we see that $S \mod p_i^{m_i}$ will also have the form Eq. \eqref{eq:M_matrix_block}, i.e., the matrix blocks $\left(S_i\right)_{13}=0_{n-k}$, $\left(S_i\right)_{23}=\left(S_i\right)_{43}=0_{k\times(n-k)}$. This proves that $S_i \in \Tnk{p_i^{m_i}}$ for all $i=1,2,\cdots, r$, which further proves that $\GM(S) \in H$. Hence the second isomorphism in Eq. \eqref{eq:Chinese_remainder_theorem_cor} is proved. 
}
\section{Proof of Eq. \eqref{eq:bound_prod}}
\label{app:sec:inequality}
We start by noting that
\begin{equation}
\label{eq:inequality1}
  1 \ \le \ \dfrac{1-{p_i}^{-2(n-j)}}{1-{p_i}^{-(n-k-j)}},
\end{equation} we get $ \left( p_i^{m_i} \right)^{\frac{(n-k)\left(n+3k+1\right)}{2}} \ \le \nStab{n}{k}{p_i^{m_i}}$ and using Eq. \eqref{eq:count_stab2} we immediately get  $d^{\frac{(n-k)\left(n+3k+1\right)}{2}} \ \le \nStab{n}{k}{d}$. For the inequality on the right, we note that $ \dfrac{1-p^{-2(n-j)}}{1-p^{-(n-k-j)}} \le \dfrac{1}{1-p^{-(n-k-j)}} \le  \dfrac{1}{1-2^{-(n-k-j)}} $, since $p \ge 2$. Thus 
\begin{equation}
    \label{eq:frac1}
 \prod_{j=0}^{n-k-1}   \dfrac{1-p^{-2(n-j)}}{1-p^{-(n-k-j)}} \ \le \ \prod_{j=0}^{n-k-1} \ \left(1 + \dfrac{2^{-(n-k-j)}}{1-2^{-(n-k-j)}}\right) \ \le \prod_{j=0}^{n-k-1} \ \left(1 + 2^{-j}\right),
\end{equation} where we used the fact that $\frac{2^{-(n-k-j)}}{1-2^{-(n-k-j)}} \le 2^{-(n-k-1-j)}$, and re-label $j\rightarrow n-k-1-j$.  and using \begin{equation} \label{eq:trick2} \prod_{j=0}^{n-k-1} \left( 1+ 2^{-j} \right) \ \le \prod_{j=0}^{\infty}  \left( 1+ 2^{-j} \right) \ = \ \phi\left( \frac{1}{2} \right),  \end{equation} where $\phi(.)$ is the Euler function\footnote{Not the Euler totient function} \cite{Apostol}, which is a special case of the $q$-Pochhammer function \cite{q_series}. Using Mathematica, this evaluates to $4.768$ or $e^{1.56}$.
\section{Why the proof of Theorem 20. of Ref. \cite{Gross2006} fails for counting $[[n,k]]_d$ modular stabilizer codes}
\label{app:sec:Gross_proof} 
We first give the  context of Theorem 20 in Ref. \cite{Gross2006}: $d=p^m$, for some positive prime integer $p$, and some positive integer $m$, and the vector space in question is $\mathbb{F}_{p^m}^{2n}$, which is defined over the field $\mathbb{F}_{p^m}$. Here $\mathbb{F}_{p^m}$ is the Galois field extension of the base field $\mathbb{F}_p \simeq \mathbb{Z}_p$. Thus this overlaps with our scenario only when $d=p$. The proof counts the number of isotropic subspaces of $\mathbb{F}_{p^m}^{2n}$. An isotropic subspace is spanned by $n-k$ linearly independent vectors $\vect{a}_1$, $\vect{a}_2$, $\cdots$, $\vect{a}_{n-k}$, with the additional property that the symplectic inner products between any pair of these vectors is zero, i.e., $\vect{a}_i^T \L \vect{a}_j=0$, for all $i,j=1,2,\cdots,n-k$. The proof in Ref. \cite{Gross2006} begins by counting the number of possible values that $\vect{a}_1$ can take in $\mathbb{F}_{p^m}^{2n}$. Since any non-zero vector in $\mathbb{F}_{p^m}^{2n}$ is a potential candidate, the count is $d^{2n}-1$. Subsequently it counts the possible values which $\vect{a}_2$ can be in $\mathbb{F}_{p^m}^{2n} \setminus \mathrm{span} \left\{ \vect{a}_1 \right\}$. Note that $\vect{a}_2$ has to satisfy the additional condition that $\vect{a}_1^T \L \vect{a}_2= 0$. Viewing the transformation $\vect{a}_2 \rightarrow \vect{a}_1^T \L \vect{a}_2$ as a linear map on $\mathbb{F}_{p^m}^{2n}$, one notes that the symplectic condition simply demands that $\vect{a}_2$ belongs to the kernel of this linear map. The rank-nullity theorem is then used to tell us that $\vect{a}_2$ belongs to a $2n-1$ dimensional subspace of $\mathbb{F}_{p^m}^{2n}$. But since $\vect{a}_1$ also belongs to the kernel, and since $\vect{a}_2$ and $\vect{a}_1$ need to be linearly independent, we need to rule out the possibility that $\vect{a}_2 = x \vect{a}_1$ for all $x \in \mathbb{F}_{p^m}$. This gives us the count for $\vect{a}_2$. The same technique is applied iteratively to give us counts for all $\vect{a}_j$ for $j=3, \cdots, n-k$. The product of these counts gives the number of distinct ordered sets of $n-k$ ``symplectic" orthogonal vectors. To obtain the number of isotropic subspaces, one needs to divide the aforementioned product by the redundancy with which each subspace is counted. It is easily seen that this redundancy is equal to the number of distinct ordered ``symplectic'' orthogonal basis any isotropic subspace has. \newline 
Coming to our scenario, we note that the first step of the proof of Theorem 20 in Ref. \cite{Gross2006} itself is incompatible with our context, since in our case we need to rule out all the non-zero linearly \emph{dependent} vectors for $\vect{a}_1$, i.e., non-zero vectors $\vect{a}_1$ which satisfy the equation $x \vect{a}_1=0$, for a non-zero $x$. When $d$ is not prime, such vectors exist. For example consider the vector $\vect{a}=(2,0,0,0,0,0)^T$ for $d=4$, $n=3$ and $k=2$. $2 \vect{a}=(4,0,0,0,0,0)^T=0$ in $\mathbb{Z}_4$. If one still manages to weed out those non-zero vectors for $\vect{a}_1$, one still encounters similar issues for counting $\vect{a}_j$ for $ j \ge 2$. According to Lemma \ref{lem:count} the total number of ordered linear independent vectors $\vect{a}_j$ whose symplectic inner products are zero should be \begin{align}
    \label{eq:count_stab_ordered}
      & \  \dfrac{\left| \Spnd{2n}{d}  \right|}{\left| B_S\left(n,k,d \right) \right| \ \left| B_A\left( n,k,d \right) \right| \  \left| \Spnd{2k}{d} \right|}, \ \mathrm{when} \ k \ge 1 \ \notag \\ 
    & \dfrac{\left| \Sp  \right|}{\left| B_S \left( n, 0, d \right) \right|}, \ \mathrm{when} \ k=0.
\end{align} This is because the redundancy, which is equal to the number of choices for ordered bases, is $\left|\GLnd{n-k}\right|$. Thus one may relate our method with the method used by Gross in Theorem 20 in Ref. \cite{Gross2006} in this way. It is possible that one may compute the expression in Eq. \eqref{eq:count_stab_ordered} by suitably modifying the counting technique in Ref. \cite{Gross2006}, but our method has the advantage of being conceptually richer. 
%


\end{document}